\documentclass{amsproc}
\usepackage{graphicx}

\newtheorem{theorem}{Theorem}[section]
\newtheorem{lemma}[theorem]{Lemma}
\newtheorem{proposition}[theorem]{Proposition}

\newtheorem{definition}[theorem]{Definition}

\begin{document}

\title[Strange attractors for Overbeck -Boussinesq model]{Strange attractors for Overbeck -Boussinesq model }

\author{Sergei Vakulenko}

\address{Institute for Mechanical Engineering Problems,
 Bolshoy pr. V. O. 61, 199178, Saint Petersburg,   and
Saint Petersburg National Research University of Information Technologies, Mechanics and Optics, 197101, St. Petersburg, Russia}


\begin{abstract}
In this paper, we consider
  dynamics defined by the Navier-Stokes equations in the Oberbeck-Boussinesq approximation in a two dimensional domain.  This model of fluid dynamics  involve fundamental physical effects: convection, and diffusion.
The main result  is as follows:
 local semiflows, induced by this problem, can  generate   all possible
 structurally stable  dynamics defined by $C^1$ smooth vector fields on  compact smooth
 manifolds  (up to an orbital topological equivalency).  To generate a prescribed dynamics, it is sufficient to 
  adjust some parameters in the  equations, namely, the viscosity coefficient,  an external heat source,  some parameters in boundary conditions
and the small perturbation of the gravitational force.

\end{abstract}

\maketitle

\section{Introduction}
\label{intro}


 The hypothesis that  turbulence can be generated by strange (chaotic) attractors was pionereed in  \cite{RT, NRT}.   In this paper,  an analytical proof of this hypothesis
for the Oberbeck-Boussinesq  (OB) approximation of Navier-Stokes equations is stated.
The OB  equations
 describes  a   model
of fluid dynamics, which involves fundamental  effects: convection and  heat transfer. Boussinesq flows are common in nature (such as atmospheric fronts, oceanic circulation), and industry.   It is known the OB approximation is very accurate for many  flows
important in applications \cite{Feir}.

  The main result of this paper can be outlined as follows. We consider the initial boundary values problem (IBVP) defined by
the Navier-Stokes equations  in the OB approximation and standard boundary conditions on a  rectangle $\Omega \subset {\mathbb R}^2$. 
 Global semiflows, induced by that IBVP, can  generate   all possible
 hyperbolic  dynamics defined by $C^1$-smooth vector fields on  finite dimensional compact smooth
 manifolds  (up to an orbital topological equivalencies). 
The well known examples of  hyperbolic dynamics
with a ''chaotic" behaviour are  the Anosov flows, the Smale A-axiom systems and  the Smale horseshoes  \cite{Ru, Katok, Holmes}.
To generate a prescribed hyperbolic dynamics, it is sufficient to
 adjust some parameters involved in the IBVP formulation, in particular, 
the viscosity coefficient  $\nu$, a spatially inhomogeneous heat source
and a small spatially inhomogeneous perturbation of the gravitational force.

The main technical tool used in the proof  is the method of realization of
vector fields (RVF) proposed by P. Pol\'a\v cik \cite{Pol1, Pol2, DaP}.
  Let us outline   the RVF  method and some
 previously obtained results.

Let us consider an  IBVP associated with a system of PDE's and 
 involving a parameter $\mathcal P$.  Assume that for each value of $\mathcal P$ that IBVP generates
a global semiflow $S^t$.    We obtain then a family $\mathcal F$ of global semiflows $S^t_{\mathcal P}$,
where each semiflow depends on the parameter   ${\mathcal P}$.
Suppose that
for an integer $n> 0$ there is an appropriate value ${\mathcal P}_n$ of the parameter ${\mathcal P}$
such that the corresponding global semiflow $S^t_{{\mathcal P}_n}$  has an $n$-
dimensional finite $C^{1}$ -smooth locally invariant manifold  ${\mathcal M}_n$.
The semiflow $S^t_{{\mathcal P}_n}$, restricted to ${\mathcal M}_n$   is
 defined  by a vector field $Q$ on ${\mathcal M}^n$. Then we say that the family $S^t_{\mathcal P}$
realizes the vector field $Q$.

By these realizations
it is shown that semiflows
associated with some special quasilinear parabolic equations in two dimensional
domains can generate complicated hyperbolic  dynamics \cite{Pol2,DaP}.
For  a large class of reaction-diffusion systems the RVF method allows us to prove 
existence of chaotic attractors \cite{Vak4}.
One can show that, for each integer $n$, semiflows induced by these systems  can realize
a dense set in the space of all $C^1$-smooth vector fields on the unit ball ${\mathcal B}^n \subset {\mathbb R}^n$ \cite{Vak4}. Therefore, such systems
generate all structurally stable
 ( persistent under
sufficiently small $C^1$- perturbations) dynamics, up to orbital topological equivalency  \cite{Ru, Katok}.
The  families of the semiflows  enjoying such property of the dense realization can be called {\em maximally  complex}.
If a family of semiflows is maximally complex, that family generates all
 hyperbolic dynamics  on finite dimensional compact smooth manifolds.
By this terminology, the main result of this paper is as follows.
The family of semiflows, associated with  the IBVP's generated by the OB equations, is maximally  complex (see Theorem \ref{maint}). 

Using the RVF method for the OB equations we  encounter  the following main  difficulty: how to reduce the OB dynamics to  a system
 of differential equations with quadratic nonlinearities
\begin{equation} \label{0.1}
 \frac{dX}{dt}={ K}(X) + { M}X + { f},\quad X \in {\mathbb R}^N, 
\end{equation}
where  $  X(t) $ is a unknown function, $X=(X_1, ..., X_N) \in {\mathbb R}^N$,  $K(X)$ is a quadratic term, $f \in {\mathbb R}^N$, the linear term $MX$ is defined by
 a $N \times N$ matrix $M$.
This reduction  is based on a construction of   locally attracting invariant manifolds  ${\mathcal M}_N$ for the OB equations.
To this end, we  choose  the IBVP parameters  in  a special way, since 
  it is impossible to prove  existence of locally attracting invariant or inertial manifolds
for the general OB equations.  We adjust  the parameter  $\mathcal P$ such that 
a linear operator $L$, that determines the linearization of the OB equations, has special spectral properties. Namely, this operator has $N$ zero eigenvalues and
all the rest spectrum of $L$ lies in the negative half plane and it is separated by a positive barrier from the imaginary axis.
 To find  the operator $L$ having such properties  is not  easy, and the construction of $L$  is a main
 technical part  of the paper.  The  further application of the RVF method to \eqref{0.1} follows works \cite{Sud, Stud, Vak4} with small modifications.

Let us outline this reduction on ${\mathcal M}_N$ in more detail. We are seeking for a solution $({\bf v}, u)$,  where $\bf v$ is the fluid velocity, $u$ is the temperature (or the impurity density), $x,y$ are horizontal and
 vertical coordinates, respectively, and  the gravitational force is directed  along $y$. Let us assume that the solutions are small perturbations of the  flow ${\bf v}={\bf 0}, u=U(y)$,
which have the form ${\bf v}=\gamma \tilde {\bf v}, u(x,y,t)=U(y) +\gamma u_1(x,y) +\gamma w(x,y,t)$, where $\tilde {\bf v}, w$ are new unknown functions, the terms $U$ and $u_1$ are adjusted in a special way and $\gamma>0$ is a small parameter.  
 The operator $L$ is defined via  $U$ and $L$ does not depend on $\gamma$.
The function $u_1$ defines the matrix ${ M}$ in equations (\ref{0.1}).  

Since $U$ depends only on $y$, we can  separate variables in the spectral problem for the operator  $L$ (this method is well known, 
  see  \cite{Drazin, GJ}).  Eigenfunctions of $L$ have the form $e_k=(\Psi_k(y)\sin(kx) ,  \Theta(y)\cos(kx))^{tr} $  with eigenvalues $\lambda_k$, where $k=1,2, ... \ $.  
  In the classical approach $U(y)$ is a linear function of $y$ \cite{Drazin, GJ}. Then  one obtains that
there are possible bifurcations, where $\lambda_k$ changes its sign at  a value ${\mathcal P}_0$ for a  $k=k_0$.  For
small $Re\ \lambda(k_0) >0$ and $\gamma>0$ the solution is $X e_{k_0}$ 
 (up to small corrections) and the magnitude $X$ can be obtained from a simple nonlinear equation for $X$.

In this paper, the main trick is as follows. 
 We take $U$ as a  fast decreasing exponent perturbed by a small polynomial, $U=C_U b^{1-s_1} \exp(-by) + \mu y P_N(y)$,
where  $\mu=b^{-s_2}$, $s_1, s_2  \in (0,1)$, $b$ is a large parameter, independent of $\gamma$  and $C_U$ is a  parameter
of the order $1$ (it does  not depend on $b$ for large $b$). The function $P_N$ is a polynomial of the degree $N$.   
For each $k$  the spectral problem for $L$ can be reduced to a nonlinear equation for  $\lambda_k$.
   In the limit $b \to \infty$, $\nu \to +\infty$, and if $P_N(y) \equiv 0$,
the equation for $\lambda_k$ has a simple limit form, which {\em  does not involve $k$}.   Furthermore, we can use the small term $\mu P_N$ to control
location of the roots $\lambda_k$ of this equation. Namely,  let us choose some $k_1, k_2, ..., k_N << b$. Under a special choice of $ P_N$,$C_U$ and some other parameters  we have $\lambda_k=0$,  $k=k_1, k_2, ..., k_N$,  whereas all others $\lambda_k$    satisfy $ Re \  \lambda_k  < - \delta(b)$, where $\delta>0$. 
When we vary $C_U$,
all $\lambda_k$  with $k=k_j$ pass through $0$  simultaneously.  In this case  there is   a bifurcation, which involves a number of the unstable modes 
$e_k$  (this effect was found  in \cite{Sud} for the Marangoni
problem).

The special spectral properties of the operator $L$ allow us to proceed
the reduction of the OB equations to (\ref{0.1})
by a quite routine procedure, which uses the well known results of invariant manifold theory \cite{He, Bates, CLu}.
This procedure shows that $K$ and $M$ depends  on the problem parameters differently, namely,  the coefficients involved in $K$
depend on  the eigenfunctions of $L$  whereas $M$ is a linear functional $M=M[u_1]$ of   $u_1(x,y)$. We show that, as
$u_1$ runs over the set of all smooth functions defined on $\Omega$,  the range of this
functional is  dense in the linear space of all $N \times N$ -matrices.
This fact  allows us to apply   the results  on quadratic  systems (\ref{0.1}) from Sect. \ref{sec:5}
and completes  the proof.  

Note that systems 
(\ref{0.1}) have important applications, in particular, in chemistry, where
 they describe bimolecular chemical reactions \cite{Zab}, and for population dynamics. 
Results  on existence of complex dynamics for (\ref{0.1}) were first obtained in \cite{Korz2} (see
 also \cite{Zab}). 
 In \cite{Stud} the RVF method  is applied to investigate dynamics of  systems (\ref{0.1}).  It is shown
that systems (\ref{0.1}) generate a maximally complex family of semiflows, where parameters
 are $ N,  M$, coefficients of the bilinear form $K$ and $f$.  
 In this paper, we are dealing with a more complicated situation, when the coefficients of the quadratic forms $K(X, X)$ are fixed. 
This difficulty is not too hard and it can be overcome  by the methods  of 
 the invariant manifold theory \cite{He,CLu, Bates} 
that allows us to reduce systems (\ref{0.1}) of large dimension to  analogous systems of smaller dimension, where the coefficients, which define  $K(X)$
can be considered as free parameters  (for more detail, see Sect.   \ref{sec:5} and \cite{Sud}).

In physical words, one can say that the interaction of the corresponding slow modes associated with the eigenfunctions $e_{k_j}$ can generate a complicated dynamics and 
spatio-temporal patterns. The relations, obtained in the paper, 
 give an analytical description of spatio-temporal patterns induced by strange attractors. 
The patterns are similar to found in \cite{Sud} for Marangoni flows, they are quasiperiodic in $x$, and have the boundary layer form, i.e., located
at the top boundary $y=0$.

The paper is organized as follows.
In the next sections   we formulate  the problem and describe the RVF method.   In Sect. \ref{sec:3a} we state the main result.
 In Sect. \ref{sec:7}   it is shown that the IBVP   is well posed
and defines a global semiflow. In the next section we introduce the operator $L$. In Sect. \ref{Spectrum}, which is a key technical part of the paper  we investigate
that operator,   and show that $L$ has  
needed spectral properties. 
In Sect. \ref{sec: 9}   we   prove existence of the finite dimensional 
 invariant manifold. In Sect. \ref{sec: 10} we check conditions, which is critically important for
the RVF method. Here we show that, for each fixed $N$,  by a  choice $u_1(x,y)$, we can obtain any prescribed  matrices ${M}$. In Sect. \ref{sec:5}   we consider quadratic systems (\ref{0.1}).    The remaining part of the 
proof is stated  in Sect. \ref{sec: 11}.

Below we use the following standard convention: all positive constants,  independent
of the  parameters $b, \gamma$ and $\nu$, are denoted by $c_i, C_j$. To diminish a formidable number of indices $i,j$,
we shall use sometimes the same indices  assuming that the constants can vary from a line to a line.

\section{ Statement of the problem}
\label{sec:2}
We consider the Oberbeck- Boussinesq of 
the Navier Stokes equations: 
\begin{equation}
  {\bf v}_t +  ({\bf v} \cdot \nabla) {\bf v} =\nu \Delta {\bf v} - \nabla p  + \kappa {\bf e}(1+\gamma g_1)  (u-u_0),
\label{OB1}
\end{equation}
\begin{equation}
     \nabla \cdot {\bf v} =0,
\label{div}
\end{equation}
\begin{equation}
   u_t +  ({\bf v} \cdot \nabla) u = \Delta u + \eta,
\label{OB2}
\end{equation}
  where  ${\bf v}=(v_1(x, y, t), v_2(x, y, t))^{tr}$,
$u=u(x,y,t), p=p(x,y, t)$ are   unknown functions
defined on $\Omega \times \{ t \ge 0 \}$,
 the domain $\Omega$ is a rectangle, 
$\Omega=[0, \pi] \times [0,h] \subset {\bf R^2}$.
Here $\bf v$ is the fluid velocity, where $v_1$ and $v_2$ are the normal and tangent velocity components,
                           $\nu$ is
the viscosity coefficient, $p$ is the pressure,
$u$ is the temperature,
$\eta(x,y)$ is a  function describing a distributed
heat source,
${\bf v} \cdot \nabla$ denotes the advection operator $v_1 \frac{\partial}{\partial x} +
v_2 \frac{\partial}{\partial y}$.  The unit vector ${\bf e}$ is directed along the vertical $y$-axis: ${\bf e}=(0,1)^{tr}$, $\kappa$ is the coefficient of thermal expansion and a
constant  $u_0$ is the reference temperature.  The term $\gamma g_1(x,y)$ is  a space inhomogeneous perturbation of the gravitational force, where  $\gamma >0$ is a small parameter, 
and $g_1(x, y)$ is a smooth function.  We assume that the non-perturbed density $\rho_0=1$.

The initial conditions are
\begin{equation}
  {\bf v}(x, y, 0)={\bf v}^0(x, y),
\quad { p}(x, y, 0)={ p}^0(x, y),
\quad u(x,y,0)=u^0(x,y).
\label{inidata}
\end{equation}

The function $u$ satisfies the boundary conditions 
\begin{equation}
  u_x(x, y, t)\vert_{x=0, \pi} =0,
\label{boundNeum}
\end{equation}
 and 
\begin{equation}
    u_y(x, y, t)\vert_{y=0} =\beta u(x, 0,t),   \quad u_y(x, y, t)\vert_{y=h} =\beta_1 u(x, h,t).
\label{boundDir}
\end{equation}

For the fluid velocity  we set conditions of the free surface at  the vertical boundaries $x=0,\pi$  :
\begin{equation}
  v_1(x, y, t)\vert_{x=0, \pi} =0,  \quad  \frac{\partial v_2(x, y, t)}{\partial x}\vert_{x=0, \pi}=0
\label{bound5x}
\end{equation}
and the no-flip condition at $y=0, y=h$:
\begin{equation}
{\bf v}(x,y,t)\vert_{y=0, h}  = {\bf 0}.
\label{Maran1}
\end{equation}

\section{RVF method}
\label{sec:3}

Before to formulate  the main theorem, let us describe
 the method of the realization of vector fields (RVF)
invented by
 P. Pol\'a\v cik (see \cite{Pol1,Pol2}).
We change slightly the original
version to adapt it for our goals.

Let us consider a family of local semiflows $S^t_{\mathcal P}$ in a fixed Banach
space $B$. Assume
 these semiflows  depend on a parameter ${\mathcal P} \in B_1$, where
$B_1$  is another Banach space.
Denote by ${\mathcal B}^n(R)$ the  ball $\{q: |q|\le R\}$ in ${\mathbb R}^n$, where $q=(q_1, q_2, ..., q_n)$ and
$ |q|^2=q_1^2 +
... + q^2_n$.  For $R=1$ we will omit the radius $R$,  ${\mathcal B}^n={\mathcal B}^n(1)$. Remind that 
a  set $M$ is said to be locally invariant in an open set $W \subset B$  under a semiflow $S^t$ in $B$ if  $M$ is a subset of $W$ and each trajectories of $S^t$ leaving $M$ simultaneously leaves
$W$.  In this paper, all $W$ are tubular neighborhoods of the balls ${\mathcal B}^n(R)$, which have small widths.
Consider a system of differential equations defined on the ball ${\mathcal B}^n$:
\begin{equation}
 \frac{dq}{dt}=Q(q),
\label{ordeq}
\end{equation}
 where
\begin{equation}
   Q \in C^1({\mathcal B}^n), \quad \sup_{q \in {\mathcal B}^n}|\nabla Q(q)| < 1.
\label{cond1}
\end{equation}
 Assume the vector field $Q$ is directed strictly
inward at the boundary $\partial {\mathcal B}^n=\{q: |q|=1 \}$:
\begin{equation}
   Q(q)\cdot q < 0 , \quad  q \in \partial {\mathcal B}^n.
\label{inward}
\end{equation}
Then system (\ref{ordeq}) defines a
global semiflow
 on ${\mathcal B}^n$. Let $\epsilon$ be a positive number.

\begin{definition} ({\bf realization of vector fields})  \label{RealVF} 
{We say that the family of local  semiflows $S^t_{\mathcal P}$  realizes the vector field $Q$ (dynamics (\ref{ordeq})) with accuracy $\epsilon$
(briefly, $\epsilon$  - realizes),
if there exists a  parameter ${\mathcal P}={\mathcal P}(Q, \epsilon, n)  \in B_1$
such that

({\bf i}) semiflow $S^t_{\mathcal P}$ has a locally invariant   in a open domain ${\mathcal W} \subset B$ and locally attracting
 manifold ${\mathcal M}_n \subset B$ diffeomorphic to the unit ball $ {\mathcal B}^n$;

({\bf ii})  this manifold is embedded into $B$  by a  map
\begin{equation}
  z = Z(q), \quad q \in {\mathcal B}^n, \quad z \in B, \quad Z \in C^{1+r}
({\mathcal B}^n),
\label{manifRVF}
\end{equation}
where $r > 0$;

({\bf iii}) the restriction of the semiflow  $S^t_{\mathcal P}$ to ${\mathcal M}_n$  is defined by the system of differential equations
\begin{equation}
  \frac{dq}{dt}=Q(q) + \tilde Q(q, {\mathcal P}), \quad Q \in C^{1}({\mathcal B}^n),
\label{reddynam}
\end{equation}
  where
\begin{equation}
	    |\tilde Q(\cdot, {\mathcal P})|_{C^1({\mathcal B}^n)} <
\epsilon.
\label{estRVF}
\end{equation}}
\end{definition}


\begin{definition}  \label{RVFmax} Let $\Phi$ be a family of vector fields $Q$, where each $Q$ is defined on a ball ${\mathcal B}^n$, positive integers $n$ may be different.  
We say that the family ${\mathcal F}$ of local semiflows $S^t_{\mathcal P}$ realizes the family $\Phi$ if for each $\epsilon>0$ and each $Q  \in \Phi$ the filed $Q$ can be 
$\epsilon$
-realized by the family ${\mathcal F}$.

We say that the family ${\mathcal F}$ is maximally dynamically complex if that family realizes all $C^1$- smooth finite dimensional fields defined on all
unit balls ${\mathcal B}^n$.

\end{definition}

\section{ Main results }
\label{sec:3a}

The  IBVP defined by  (\ref{OB1}) -(\ref{Maran1}) involves the coefficients
$ \nu,  h, \gamma, \beta, \beta_1$ and the functions $\eta(x,y), g_1(x,y)$.
We set
${\mathcal P}=\{h, \nu, \gamma, \beta, \beta_1, u_0, \eta(\cdot,\cdot), g_1(\cdot,\cdot) \}$. The main result  is as follows:

\begin{theorem} \label{maint}
{ The family  of the semiflows defined by IBVP  (\ref{OB1}) -(\ref{Maran1})   is maximally dynamically complex, that is,
for each integer $n$, each $\epsilon > 0$  and each vector field
$Q$ satisfying (\ref{cond1}) and (\ref{inward}), there exists a value of the parameter
${\mathcal P}(Q,\epsilon)$  such that
IBVP  (\ref{OB1}) -(\ref{Maran1}) defines a  semiflow $S^t_{\mathcal P}$, which $\epsilon$-realizes
the vector field $Q$}.
\end{theorem}

Persistence of hyperbolic sets \cite{Katok} and some  additional arguments then implies the following corollary.

\begin{theorem} \label{maint2}{ The family of semiflows  $S^t_{\mathcal P}$ induced  by  IBVP  (\ref{OB1}) -(\ref{Maran1})
  generates
all  (up to  orbital topological equivalencies)   hyperbolic dynamics on compact invariant hyperbolic sets  defined by
$C^1$-smooth vector fields  on  finite dimensional smooth compact manifolds}.
\end{theorem}

In particular, we obtain that the OB dynamics can generate  Smale axiom A flows,
Ruelle-Takens attractors \cite{RT, NRT} and the Anosov flows.

\section{Existence and uniqueness}
\label{sec:7}

\subsection{ Function spaces and embeddings}
\label{sec: 7.1}

We  use the standard Hilbert  spaces \cite{He}.
 We denote by $H$ the closure in $L_2(\Omega, {\mathbb R}^2)$ of the set  of $C^1$-smooth  vector valued functions ${\bf v}=(v_1, v_2)$ such that 
$\nabla \cdot {\bf v}=0$ and satisfying boundary conditions \eqref{bound5x}, \eqref{Maran1}.

The space $H$ is equipped  by norms $|| \ ||$, where $||{\bf v}||^2=\langle v_1, v_1\rangle + \langle v_2, v_2\rangle$ 
 and $\langle, \rangle$ is the inner product in $L_2(\Omega)$  defined by 
\begin{equation}
 \langle f, g \rangle =\int_0^h \int_0^{\pi} f(x,y) g(x,y) dx dy.
\label{inprod}
\end{equation}

Let us  denote by  $H_{\alpha}$ the fractional spaces
\begin{equation}
H_{\alpha}= \{ {\bf v} \in H:  ||{\bf v}||_{\alpha} =||(I-\Delta_D)^{\alpha} \bf v|| < \infty \},
\label{Hal}
\end{equation}
here $\Delta_D$ is the Laplace operator  with the  domain corresponding to   boundary conditions
\eqref{bound5x}, \eqref{Maran1}.
 Let  $\tilde H_{\alpha}$  be another fractional space associated with $L_2(\Omega)$:
\begin{equation}
\tilde H_{\alpha}= \{ u  \in L_2(\Omega):  ||u||_{\alpha} =||(I-\Delta_N)^{\alpha} u|| < \infty \},
\label{Bp}
\end{equation}
where $\Delta_N$ is the Laplace operator with the  domain corresponding to the boundary conditions \eqref{boundNeum}, \eqref{boundDir}.
Below we  omit  the indices $N, D$.

The  Sobolev embeddings
\begin{equation}
\tilde H_{\alpha} \subset C^{s} (\Omega), \quad 0 \le s < 2(\alpha- 1/2),
\label{emb1}
\end{equation}
and
\begin{equation}
\tilde H_{\alpha} \subset L_q (\Omega), \quad 1/q > 1/2- \alpha, \ q \ge 2
\label{emb2}
\end{equation}
will be used below, and
analogous embeddings for $H_{\alpha}$. 
We consider IBVP  (\ref{OB1})-(\ref{Maran1})  in the phase space
${\mathcal H}=H \times \tilde H$.

In coming subsection our aim is to prove that the IBVP  (\ref{OB1})-(\ref{Maran1}) defines a global semiflow.


\subsection{Evolution equations}

To show local existence of solutions we  use the standard semigroup methods.  Let ${\mathbb P}$ be the Leray projection (see \cite{Temam}) and ${\bf v} \in 
\mathbb{P}  H$.  
Then we can rewrite our IBVP as an evolution equation for the pair $z=({\bf v}, u)^{tr}$: 
\begin{equation}  \label{equ}
z_t =  A z  +   F(z),
\end{equation}
where 
$$
A=( \nu {\mathbb P} \Delta_D,    \Delta_N)^{tr},  \quad F=(F_1, F_2)^{tr},
$$
$$
F_1=  {\mathbb P} \Big(-({\bf v}\cdot \nabla) {\bf v}  +    \kappa {\bf e}(1+  \gamma g_1)(u -u_0 ) \Big),
$$
$$
F_2= - ({ \bf v}\cdot \nabla) u  +  \eta.  
$$
By \eqref{emb2} we observe that 
\begin{equation} \label{embd}
|| ({\bf v} \cdot \nabla) {\bf v}|| \le |{\bf v}|_{\infty} ||{\bf v}||_{\alpha}
\end{equation}
where $\alpha \in (1/2, 1)$ and
$|f|_{\infty}$  denotes the supremum norm:
$$
|f|_{\infty}=\sup_{x, y \in \Omega}| f(x,y)|, \quad |{\bf v}|_{\infty}=|v_1|_{\infty} + |v_2|_{\infty}.
$$
Estimate \eqref{embd} and an analogous estimate for  $|| ({\bf v} \cdot \nabla) u||$ show  that for $\alpha >1/2$ the map 
 $F$ is a  bounded $C^1$- map from a bounded domain in ${\mathcal H}_{\alpha}=H_{\alpha} \times \tilde H_{\alpha}$ to
$\mathcal H$ \cite{He}.  This fact implies a local existence and uniqueness of solutions of \eqref{equ}. So, eq. (\ref{equ}) 
 defines a local semiflow in $\mathcal H$ .

\begin{proposition} \label{Prop1}
Let  $\beta >0$ and $\beta > \beta_1$. Then the  IBVP   defined by  (\ref{OB1}) -(\ref{Maran1}) generates a global semiflow in $\mathcal H$.
\end{proposition}

\begin{proof}
Global existence and boundedness of solutions can be derived by the differential inequalities
\begin{equation} \label{Diff1}
\frac{1}{2} \frac{d||{\bf v}||^2}{dt}  \le - \nu ||\nabla {\bf v}||^2 + \kappa ( 1 +  \gamma |g_1|_{\infty}) ||v_2|| ||u-u_0||,  
\end{equation}
and
\begin{equation} \label{Diff2}
\frac{1}{2} \frac{d||u||^2}{dt}  \le -  ||\nabla u||^2 +   ||\eta|| ||u|| + c_1||v_2|| ||u||  + I_{\beta}(u),  
\end{equation}
where
\begin{equation} \label{Ibeta}
I_{\beta}(u)=\beta_1 \int_0^{\pi}  u^2(x,h) dx -\beta \int_0^{\pi}  u^2(x,0) dx.   
\end{equation}
For $\beta >0$  one has 
\begin{equation} \label{Ibeta1}
I_{\beta} \le  \bar \beta_1 \int_0^{\pi}  (u^2(x,h) - u^2(x,0) )dx =\frac{\bar \beta_1}{2} \int_0^h  \int_0^{\pi} u_y  u  dxdy,
\end{equation}
where $\bar \beta_1=\max\{\beta_1, 0\}$.
Thus for each $a >0$ 
\begin{equation} \label{Ibeta2}
I_{\beta} \le \frac{\bar \beta_1}{2} ||u_y|| ||u|| \le \frac{\bar \beta_1}{4} (a||\nabla u||^2 +  a^{-1}||u||^2).
\end{equation}
Choosing an appropriate $a$ we see that inequalities  \eqref{Diff1}, \eqref{Diff2} and estimate \eqref{Ibeta2}
 lead to by a priori estimate
$|| u(\cdot, \cdot, t) ||  + || {\bf v}(\cdot, \cdot, t) ||< c_2 \exp(c_3 t) $ for all $t  \ge 0$. 
Thus, we can conclude that eq. (\ref{equ}) defines a  global  semiflow in $\mathcal H$. \end{proof}

\section{ Linearization of the problem }
\label{sec: 7.4}

First we follow the standard approach developed for the Rayleigh- B\'enard  convection
\cite{Drazin,GJ} but with small modifications. 

Let 
$U$(y) be a $C^{\infty}$-smooth function of $y \in [0,h]$ such that
\begin{equation}  \label{Ubc}
\frac{dU}{dy}\vert_{y=0}=\beta U(0),  
\end{equation}
 \begin{equation}  \label{Ubc1}
 \frac{dU}{dy}\vert_{y=h}=\beta_1 U(h).
\end{equation}
 For sufficiently small $\gamma$ and $u_0$ such that $u_0 > |U|_{\infty}$ we set
\begin{equation}  \label{Uu1}
u_1=-g_1(U- u_0) (1+ \gamma g_1)^{-1}.
\end{equation} 
We suppose that 
\begin{equation}  \label{gbc} 
\frac{\partial g_1(x, y)}{\partial x}\vert_{x=0, \pi}=0 \quad \forall \  y \in [0,h], 
\end{equation}
and
\begin{equation}  \label{gbc1}
g_1(x, y)\vert_{y=0, h}=0  \quad  \quad \frac{\partial g_1(x, y)}{\partial y}\vert_{y=0, h}=0 \ \forall \ x \in [0, \pi].
\end{equation}
Then 
\begin{equation}  \label{ubc2} 
\frac{\partial u_1(x, y)}{\partial x}\vert_{x=0, \pi}=0 \quad \forall \  y \in [0,h], 
\end{equation}
and
\begin{equation}  \label{ubc3}
\frac{\partial u_1(x, y)}{\partial y}\vert_{y=0}=\beta u_1(x,0) \quad \forall \  x \in [0,\pi],
\end{equation}
\begin{equation}  \label{ubc4}
\frac{\partial u_1(x, y)}{\partial y}\vert_{y=h}=\beta_1 u_1(x,h) \quad \forall \  x \in [0,\pi].
\end{equation}
Assume that
$$
\eta=\eta_0 + \gamma^2 \eta_1, \quad \eta_0=-\Delta (U + \gamma u_1),
$$
where $\eta_1$ is a smooth function, which will be considered as a parameter.
Let us represent $u$ and ${\bf v}$ as
\begin{equation}
u=U + \gamma u_1 +\gamma w, \quad {\bf v}=\gamma \tilde {\bf v},
\label{transform1}
\end{equation}
where $w, \tilde {\bf v}$ a new unknown functions. Taking into account that the Leary projection of a gradient field is zero, and using substitution \eqref{transform1}
we note that eq. (\ref{equ})  can be rewritten as 
\begin{equation}
 \tilde {\bf v}_t =   {\mathbb P} \Big (\nu \Delta \tilde {\bf v}    - \gamma
 ( \tilde {\bf v} \cdot \nabla) \tilde {\bf v}  +   \kappa {\bf e} w (1+ \gamma g_1) \Big),
\label{eveq10}
\end{equation}
\begin{equation}
w_t= \Delta w   -  \tilde v_2 U_y - \gamma (\tilde {\bf v} \cdot\nabla) (u_1 + w) + \gamma \eta_1.
\label{eveq11}
\end{equation}
Due to \eqref{Ubc}, \eqref{Ubc1}, \eqref{ubc2},  \eqref{ubc3},  and \eqref{ubc4} the new unknowns $\tilde {\bf v}$ and $w$ satisfies the same homogeneous boundary conditions
that $\bf v$ and $u$. 

Removing the terms  of  the order $\gamma$  in (\ref{eveq10}), (\ref{eveq11}),
we obtain the linear operator
\begin{equation}
 L\tilde z=(\bar L_1 \tilde  z, \bar L_2 \tilde  z)^{tr},   \quad \tilde z=(\tilde {\bf v}, w)^{tr}
\label{Lop}
\end{equation}
 where the  operators $\bar L_k$ are defined by
\begin{equation}
\bar L_1 {\bf v}= {\mathbb P} (\nu \Delta \tilde {\bf v}   +\kappa {\bf e} w), \quad \bar L_2 w=\Delta w-  \tilde  v_2 U_y.
\label {A1}
\end{equation}
The spectral problem for the operator $L$ has the form 
\begin{equation}
\lambda  {\bf v} = {\mathbb P} (\nu \Delta {\bf v}  +\kappa{\bf e} w) ,
\label {Sp1}
\end{equation}
\begin{equation}
\lambda w=\Delta w  -  v_2 U_y, 
\label {Sp2}
\end{equation}
where ${\bf v}$ satisfies boundary conditions   \eqref{bound5x} and \eqref{Maran1}, and $w$ satisfies the boundary conditions 
\begin{equation}
  w_x(x, y)\vert_{x=0, \pi} =0, \quad \forall \  y \in [0, h],
\label{boundNeumSp}
\end{equation}
\begin{equation}
    w_y(x, y)\vert_{y=0} =\beta w(x, 0),   \quad w_y(x, y)\vert_{y=h} =\beta_1 w(x, h),  \quad  \forall \ x \in [0, \pi].
\label{boundDirSp}
\end{equation}

   This spectral problem  is investigated in coming sections but first we consider some  properties of the operator $L$.

\subsection{ Properties of $L$ }
\label{sec: 7.5}

In order to apply the standard technique \cite{He}, first let us  show  that the operator $L$
is sectorial.

\begin{lemma} \label{5.4} { $L$ is a sectorial operator}.
\end{lemma}

\begin{proof}
 We use the following  result \cite{Kato, He}: if $L^{(0)}$ is a self adjoint operator in a Banach space $X$,
$L^{(0)}: X \to X$ and
$B$ is a linear operator, $B: X \to X$ such that $Dom \ L^{(0)} \subset Dom \ B$ and for all $\rho \in Dom \ L^{(0)}$
\begin{equation}
||B \rho || \le \sigma || L^{(0)} \rho || + K(\sigma) ||\rho||
\label{sect}
\end{equation}
for  $0 < \sigma < 1$ and a constant $K(\sigma)>0$, then $L^{(0)} + B$ also is a sectorial operator.

  Let us define the unperturbed operator $L^{(0)}$ by the relations
$$
  \bar L_1^{(0)} ({\bf v},w)^{tr}=  \nu {\mathbb P}  \Delta {\bf v} ,  \quad \bar L_2^{(0)} ({\bf v}, w)^{tr} = \Delta w,
$$
where $\rho=({\bf v},w)^{tr} \in {\mathcal H}$.
The  operator $L^{(0)}$
is self-adjoint in the space $\mathcal H$, its spectrum is discrete
 and lies in the interval $(-\infty, 0)$.
  Therefore,  $-L^{(0)}$ is a sectorial.
The operator $B$ is given then by
$$
 B ({\bf v}, w)^{tr} =(\kappa w,  -v_2  U_y)^{tr}.
$$
It is clear that estimate (\ref{sect}) is satisfied. \end{proof}

\begin{lemma} \label{resolv} {For some $\lambda >0$ and positive $\beta$ such that $\beta >\beta_1$ the resolvent $(L- \lambda)^{-1}$ is a compact operator from $\mathcal H $ to $\mathcal H$ .
}
\end{lemma}

\begin{proof} Consider equations  
\begin{equation}
\lambda  {\bf v} = {\mathbb P} (\nu \Delta {\bf v}  +\kappa{\bf e} w + {\bf g}) ,
\label {Sp1K}
\end{equation}
\begin{equation}
\lambda w=\Delta w -  v_2 U_y + f,
\label {Sp2K}
\end{equation}
where ${\bf g}, f$ lie in $H$ and $\tilde H$, respectively. These equations imply  the estimates
$$
\lambda ||{\bf v}||^2 + \nu ||\nabla {\bf v}||^2 \le\kappa ||w|| || {\bf v}|| +   ||{\bf g}|| || {\bf v}||,
$$
$$
\lambda ||w||^2 +  ||\nabla w||^2 \le c_0 ||w|| || {\bf v}|| +   ||{ f}|| || { w}|| +  I_{\beta}(w),
$$
where $c_0 >0$ is independent of $\lambda$ and $I_{\beta}(w)$ is defined by \eqref{Ibeta}.
For sufficiently large positive $\lambda \in {\mathbb R}$ the above inequalities and  \eqref{Ibeta2} imply
$$
||\nabla w|| + \nu ||\nabla {\bf v}|| \le c_1  (||f|| + ||{\bf g}||).
$$
Consequently, $L-\lambda$ is invertible and $(L-\lambda)^{-1}$ is a compact operator. \end{proof}

 According to  (see \cite{Kato}, Ch. III, Theorem 6.29)  the last  lemma implies that the spectrum of $L$ is discrete (consists of isolated eigenvalues), each
eigenvalue has a finite multiplicity $n(\lambda)$,
and the resolvent $(L-\lambda)^{-1}$ is a compact operator for all $\lambda$, where $(L-\lambda)^{-1}$ is bounded.  We investigate the spectrum of $L$ in the next section.

\section{Spectrum of  linear operator  L} \label{Spectrum}

To study the spectral problem for $L$, we use 
the stream function-vorticity   reformulation of the Navier Stokes equations  \cite{Chorin}.
The velocity  $\bf v$ can be expressed
 via the stream
function $\psi(x,y)$  by the relations
$v_1=\psi_y, v_2=-\psi_x$.  Given a ${\bf v}$, the function $\psi$ can be found  by the relation
\begin{equation} \label{psiv1}
\psi(x,y)=-\int_y^h v_1(x, s) ds.
\end{equation}
As a result of the standard transformations, eqs. (\ref{Sp1}), (\ref{Sp2})
  take the form  
\begin{equation}
\lambda \Delta \psi =\nu \Delta^2  \psi  -\kappa w_x,
\label{OBEstream2}
\end{equation}
\begin{equation}
\lambda w =\Delta w + \psi_x U_y.
\label{heat1}
\end{equation}
We obtain the following boundary conditions  for $\psi$: 
\begin{equation}
 {\psi} (x,y,\lambda)\vert_{x=0, \pi}   = \Delta \psi(x, y, \lambda)\vert_{x=0,\pi}=0,
\label{boundstream2x}
\end{equation}
\begin{equation}
 {\psi}_x (x,y, \lambda)\vert_{y=0, h} =  \psi_y(x, y, \lambda)\vert_{y=0, h}=0.
\label{Maran2}
\end{equation}

\subsection{ Some preliminaries} \label{sec:8.1}

Let us consider  the spectral problem defined by \eqref{OBEstream2},  \eqref{heat1},  \eqref{boundstream2x},\eqref{Maran2},   \eqref{boundNeumSp}   and  \eqref{boundDirSp}.
We  seek eigenfunctions $e(x,y,\lambda)=(\psi, w)^{tr}$ with eigenvalues $\lambda \in {\mathbb C}_{1/2}$, where 
 ${\mathbb C}_a$ denotes the half-plane 
\begin{equation}
{\mathbb C}_a= \{ \lambda \in {\mathbb C}:   Re \ \lambda > -a \}.
\label{Ck}
\end{equation}
In fact, we are interested in $\lambda \in {\mathbb C}_{1/2}$ because for small $\gamma$ only the eigenfunctions with  the eigenvalues $\lambda \in {\mathbb C}_a$, where $a >> \gamma$,
are involved in the construction of  the locally invariant manifold ${\mathcal M}_N$. 

Since $U=U(y)$  is independent of $x$,  
we seek the eigenfunctions  in the form 
\begin{equation}
   \quad \quad w(x,y, \lambda)= w_k(y, \lambda) \cos(kx),
\label{wF}
\end{equation}
  \begin{equation}
   \psi(x,y, \lambda)= \psi_k(y, \lambda)  \sin(kx), 
\label{psiF}
\end{equation}
where $k$ are positive integers,  $k \in {\mathbb N}=\{1,2,...,\}$. 
Let us introduce the operator $L_k=D_y^2 -k^2$.
Moreover,  to simplify formulas,  we set $\kappa=\nu$.
Then for  $\psi_k$ and $w_k$ one obtains the following 
boundary value problem on $[0,h]$:
\begin{equation} \label {BVK1}
 \lambda_k \nu^{-1} L_k \psi_k=L_k^2 \psi_k +  k^2 w_k,  
\end{equation}
\begin{equation} \label {BVK1b}
 \lambda_k w_k =L_k w_k -  U_y \psi_k,   
\end{equation}
\begin{equation} \label{Psik}
\psi_k(0)=\psi_k(h)=\frac{d\psi_k}{dy}\vert_{y=0,h}=0,
\end{equation}
\begin{equation} \label{uuk}
\frac{dw_k(y)}{dy}\vert_{y=0}=\beta w_k(0),  \quad \frac{dw_k(y)}{dy}\vert_{y=h}=\beta_1 w_k(h).
\end{equation}

Let us denote $\bar k=\sqrt{k^2 +\lambda_k}$.
We can  suppose, without loss of generality, that $Re \  \bar k >0$ for $\lambda_k \in  {\mathbb C}_{1/2}$,   since  $\bar k$ is involved in  eq. \eqref{BVK1b} only via $\bar k^2$.

\subsection{Choice of profile $U$ and parameters}

To simplify
the spectral problem, let us introduce is a large parameter  $b >0 $ and 
 assume that 
\begin{equation} \label{hviscos}
\nu > b^{10}, \quad  h=-10 \log b. 
\end{equation}
Moreover,  we set
\begin{equation} \label{rbb}
\beta= r  b, \quad  r=b^{-s_0},
\end{equation}
where the value $s_0 \in (0,1)$ will be precise below.
The key trick is the following choice of $U$:  
\begin{equation}  \label{Uy}
U(y)=\bar C_U + \int_0^y \Big(C_U r b^4 \exp(-bs) + \mu s P_N(s)  \Big)ds,
\end{equation}
where
\begin{equation} \label{mua}
 \mu = b^{-s_2},  \quad   s_2 \in (0,1), \quad 
\end{equation}
$P_N(y)$ is a polynomial of degree $N$ and $C_U, \bar C_U \ne 0$ are coefficients.  We shall precise the value of $C_U$ in the end of this section, where it will be shown 
that $C_U$ does not depend on $b$ as $b \to \infty$. To satisfy condition \eqref{Ubc}, we set
\begin{equation}  \label{Uyc}
\bar C_U =\beta^{-1}  C_U r b^4 =  C_U b^3. 
\end{equation}

We adjust $\beta_1$ from condition \eqref{Ubc1} and relation \eqref{Uy} and as a result, one obtains
\begin{equation}  \label{Uyb}
\beta_1=\big(C_U r b^4 \exp(-bh) + \mu h P_N(h) \big)B^{-1},
\end{equation}
where 
$$
B=\Big(\bar C_U + \int_0^h (C_U r b^4 \exp(-bs) + \mu s P_N(s) )ds \Big).
$$
Note that $\beta >0$ and for large $b$ one has $\beta_1 < \beta$, therefore, $\beta$ and $\beta_1$ satisfy the conditions of Prop. \ref{Prop1} and Lemma  \ref{resolv}.

\subsection{ Main result on spectrum of operator $L$}
\label{sec: 8.2}

Let us formulate the  assertion. 

\begin{proposition} \label{6.2}
{ Let   \eqref{hviscos}-\eqref{Uyb} hold,
 $N$ be a positive integer  and ${\mathcal K}_N=\{k_1, ..., k_N \}  \subset {\bf Z}_+$.
Then 
there exists a  polynomial $P_{N}(y)$ such that for  sufficiently large $b$ 
the eigenvalues  $\lambda_k$ of BVP (\ref{BVK1})-(\ref{uuk})  satisfy
\begin{equation}
\lambda_k=0  \quad    k \in {\mathcal K}_N, 
\label{Spec0}
\end{equation}
\begin{equation}
  Re \ \lambda_k < - C_N b^{-c_N}  \quad  k \notin {\mathcal K}_N,  
  \label{Spec}
\end{equation}
where positive $C_N, c_N$ are uniform in $b$ as $b \to +\infty$.
}
\end{proposition}

The plan of the proof is as follows.  We  show that the values $\lambda_k \in {\mathbb C}_{1/2}$ only under the condition $k < c_1 b^{s_1}$, where $s_1 \in (0,1)$. 
That result allows us to
find an asymptotics for the eigenfunctions, which is valid for $k/b << 1$. For $\lambda_k$ we obtain a nonlinear equation. 
The asymptotics of eigenfunctions and the property $k/b << 1$ allows us to simplify this  equation for $\lambda_k$.  By a variable rescaling 
we show that this equation is a small perturbation of a simple cubic one.   As a result, for small  $k/b$  that equation can be investigated by a perturbation technique.

First we prove  a series of  auxiliary  assertions. Consider 
the Green function $\Gamma_{\bar k}(y, y_0)$   of the operator $L_{\bar k}$ defined by 
the equation 
$$
L_{\bar k} \Gamma_{\bar k} =\delta(y-y_0)
$$
and  the boundary conditions  
$$
\frac{d\Gamma_{\bar k}}{dy}(y,y_0)\vert_{y=0}=\beta \Gamma_{\bar k}(0,y_0), 
$$
$$
\frac{d\Gamma_{\bar k}}{dy}(y,y_0)\vert_{y=h}=\beta_1 \Gamma_{\bar k}(h, y_0).
$$

\begin{lemma} \label{LG1} { Let $Re \ \bar k > 0$. Then $\Gamma_{\bar k}$
satisfies the estimate
\begin{equation} \label{LL4}
 |\Gamma_{\bar k}(y, y_0) - \bar \Gamma_{\bar k}(y, y_0)| < C_0 \exp(- Re \ \bar k (|h-y| + |h-y_0|)),  
\end{equation}
where $C_0 >0$ is a constant and $\bar \Gamma_{\bar k}$ is defined by  
\begin{equation} \label{LL4a}
 \bar \Gamma_{\bar k}(y, y_0)= \frac{  \exp(- \bar k y_0) (\sinh(\bar k y)+ \bar  k\beta^{-1} \cosh(\bar k y))}{\bar k (1 + \bar k \beta^{-1})},  \quad y < y_0,   
\end{equation}
and
\begin{equation} \label{LL4b}
 \bar \Gamma_{\bar k}(y, y_0)= \frac{  \exp(- \bar k y) (\sinh(\bar k y_0)+ \bar  k\beta^{-1} \cosh(\bar k y_0))}{\bar k (1 + \bar k \beta^{-1})}  \quad y \ge y_0.   
\end{equation}
}
\end{lemma}

\begin{proof} Let us represent $\Gamma_{\bar k,h}$ as a  sum $\Gamma_{\bar k,h}=\bar \Gamma_{\bar k} (y,y_0) + \tilde \Gamma_{\bar k}(y,y_0)$,
where $\bar \Gamma_{\bar k} (y,y_0)$ is the Green function of the operator $L_k$ on $[0, +\infty)$ under boundary conditions
$$
\frac{d\bar \Gamma_{\bar k}}{dy}(y,y_0)\vert_{y=0}=\beta \Gamma_{\bar k}(0,y_0),  \quad \lim_{y \to +\infty} {\bar \Gamma_{\bar k}}(y, y_0)=0.
$$

Then $\bar \Gamma$ is defined by \eqref{LL4a},  \eqref{LL4b} and 
 $\tilde \Gamma_{\bar k}$ is the solution of the following boundary value problem:
\begin{equation} \label{gbp}
L_{\bar k} \tilde \Gamma_{\bar k}=0, 
\end{equation}
\begin{equation} \label{gbp1}
\frac{d\tilde \Gamma_{\bar k}(y}{dy}\vert_{y=0}=\beta \tilde \Gamma_{\bar k}(0, y_0), 
\end{equation}
\begin{equation} \label{gbp2}
\frac{d\tilde \Gamma_{\bar k}(y}{dy}\vert_{y=h}=\beta_1 \tilde \Gamma_{\bar k}(h, y_0)  +\beta_2,      \quad \beta_2=\beta_1 \bar \Gamma_k(h, y_0)  - \frac{d\bar \Gamma_k(y, y_0)}{dy}\vert_{y=h}.
\end{equation}

Note that $|\bar \Gamma_{\bar k} (h,y_0)| < c |\bar k^{-1}|\exp(-Re \ \bar k|h-y_0|)$. Therefore,  $$\beta_2 < c_1 |\bar k^{-1}|\exp(-Re \ \bar k|h-y_0|).$$
Resolving the BVP defined by \eqref{gbp},\eqref{gbp1}, \eqref{gbp2} and taking into account the above estimate for $\beta_2$
we see that 
$|\tilde \Gamma_{\bar k}| < C_0\exp(- Re \ \bar k (|h-y| + |h-y_0|))$.
\end{proof}

Roughly speaking   Lemma \ref{LG1} asserts that for large $h$ the Green function $\Gamma_{\bar k}$ consists of two terms, the first
one can be computed explicitly and  the second one is a exponentially decreasing 
boundary layer term. Such a structure simplifies the analysis of the spectral problem. All terms induced by the boundary layers are negligible as $b \to +\infty$ due to our choice \eqref{hviscos} of $h$. Note that  for $\lambda \in {\mathbb C}_{1/2}$
$$
Re \ \bar k(\lambda) >  \sqrt{k^2 -  1/2} > \frac{k}{2}.
 $$
Thus estimate \eqref{LL4} implies
\begin{equation} \label{LL4a}
 |\Gamma_{\bar k}(y, y_0) - \bar \Gamma_{\bar k}(y, y_0)| < C_0 \exp(- \frac{ k (|h-y| + |h-y_0|)}{2}).
\end{equation}

\begin{lemma} \label{int}
{For  $m=0,1,2,3$ and $\lambda \in {\mathbb C}_{1/2}$ the solution $\psi_k$ of eq. \eqref{BVK1} satisfies
\begin{equation} \label{psiestm}
|D_y^m \psi_k|_{\infty} \le c_m k^{m-2} h |w_k|_{\infty}
\end{equation} 
and 
for  $m=0,1, 2$ the solution $w_k$ of eq. \eqref{BVK1b} satisfies
\begin{equation} \label{westm}
|D_y^m w_k|_{\infty} \le c_m k^{m-1}|\bar k|^{-1} |\psi_k|_{\infty} \sup_{y \in [0,h]}|U_y|.
\end{equation} 
}
\end{lemma}

\begin{proof}
To prove \eqref{westm} we use Lemma  \ref{LG1} and \eqref{LL4a}, which show that 
\begin{equation} \label{Lnu}
\int_{0}^{h} | D_y^m \Gamma_{\bar k}(y, y_0)| dy_0  < C_m k^{m-1} |\bar k|^{-1}, \quad m=0,1.
\end{equation} 
These estimates imply \eqref{westm} for $m=0,1$. For $m=2$ estimate\eqref{westm} follows from eq. \eqref{BVK1} and \eqref{Lnu}.

To prove \eqref{psiestm}, we use  the relation  
\begin{equation} \label{psich1}
||D_y^2 \psi_k||^2 + (2k^2 + \frac{Re \ \lambda}{\nu}) || D_y \psi_k||^2 + (k^4 +  \frac{k^2 Re \ \lambda}{\nu})||\psi_k||^2 =k^2 Re \ \langle \psi_k^*, w_k \rangle,
\end{equation} 
which follows from \eqref{BVK1} and where
$\langle f, g \rangle =\int_0^h fg dy$,  $\psi_k^*$ is complex conjugate to $\psi_k$, and $||f||^2 =\langle f , f^* \rangle$. 

For $\lambda \in {\mathbb C}_{1/2}$ \eqref{psich1} implies
\begin{equation} \label{psich2}
||D_y^m \psi_k||  \le  c_m k^{m-2} ||w_k|| \le c_m \sqrt{h}k^{m-2} |w_k|_{\infty}
\end{equation} 
for $m=0,1,2$. Using \eqref{psich2} and eq.  \eqref{BVK1} we extend \eqref{psich2} on the case $m=4$ and 
by $||D_y^3 \psi_k||^2 \ge ||D_y^2 \psi_k|| ||D_y^4 \psi_k||$ we obtain  \eqref{psich2} for  $m=3$.

Now the Sobolev embeddings 
$$
|D_y^m \psi_k|_{\infty}  \le c_m \sqrt{h} ||D_y^{m+1}  \psi_k||, \quad m=0,1,2,3 
$$
lead to \eqref{psiestm}. \end{proof}

\begin{lemma} \label{L2N} { If $Re  \ \lambda_k > -1/2$,  and $\mu, s_2, s_3$ are defined by \eqref{mua} 
 and a nontrivial 
solution  of  BVP (\ref{BVK1})-(\ref{uuk})  exists. Then  
\begin{equation} \label{LL2}
|\bar k| <  c_1  h r b=c_1 h b^{1 -s_0}, 
\end{equation}
where $c_1>0$ is a constant independent of $b,k$.}
\end{lemma}

\begin{proof} 
One has  $2|\psi_k(y)|\le   y^2 |D_y^2 \psi|_{\infty}$.  
That estimate and eq. \eqref{BVK1b} imply that
\begin{equation} \label{LL2p1b}
|w_k|_{\infty} \le c_2 |D_y^2 \psi_k|_{\infty} \sup_{y \in [0, h]} I_2(y),  
\end{equation} 
where
\begin{equation} \label{Imm}
I_m(y)=\int_0^h
\Big |\Gamma_{\bar k}(y, y_0) \big(C_U r b^4  \exp(- by_0)  +  \mu y_0 P_N(y_0)  \big) \Big| y_0^m dy_0.  
\end{equation}
By Lemma \ref{LG1} and \eqref{hviscos} we find that for sufficiently large $b$ one has
$$|I_2 |< c_3 |\bar k|^{-1}(rb +  b^{-s_3}), $$ 
where 
\begin{equation} \label{s4}
s_3=s_2/2.
\end{equation}
According to Lemma  \ref{int}  
\begin{equation} \label{LL2p1}
|D_y^2 \psi_k|_{\infty} \le c_0 h |w_k|_{\infty}.
\end{equation} 
The above estimate of $|I_2|$ and  \eqref{LL2p1}, \eqref{LL2p1b}  imply 
\begin{equation} \label{LL2p2}
|w_k|_{\infty} \le c_4 h (rb + b^{-s_3})  |\bar k|^{-1} |w_k|_{\infty}   
\end{equation} 
that entails \eqref{LL2}.\end{proof}

\begin{lemma} \label{Ldps} { If $Re \ \lambda_k > -1/2$,  
 then  for sufficiently large $b >0$ a nontrivial 
solution  of  problem (\ref{BVK1})-(\ref{uuk})  exists only under condition  
\begin{equation} \label{LL2p}
 \rho_2=\frac{1}{2} \frac{d^2 \psi_k(y)}{ dy^2}\vert_{y=0} \ne 0.
\end{equation}}
\end{lemma}
 
\begin{proof} 
The proof is analogous to the previous one. If $\rho_2=0$, then  $6|\psi_k(y)|\le   y^3 |D_y^3 \psi|_{\infty}$.  
This estimate and eq. \eqref{BVK1b} imply that
\begin{equation} \label{LL2p4}
|w_k|_{\infty} \le C_3 |D_y^3 \psi_k|_{\infty} \sup_{y \in [0, h]} I_3(y),  
\end{equation} 
where $I_3$ is defined by \eqref{Imm}.
As above  one has
$$|I_3| < c_4 |\bar k|^{-1}(r +  b^{-s_3}).$$ 
This estimate and the inequality    
\begin{equation} \label{LL2p3}
|D_y^3 \psi_k|_{\infty} \le c_1 k h |w_k|_{\infty},
\end{equation} 
that follows from Lemma  \ref{int},    imply 
\begin{equation} \label{LL2p2}
|w_k|_{\infty} \le c_5  h(r+ b^{-s_3}) k |\bar k|^{-1} |w_k|_{\infty}.   
\end{equation} 
For large $b$ one has $c_2 h (r+ b^{-s_3})  < 1$, thus $|w_k|_{\infty}=0$.\end{proof}

The next step is to find an asymptotics for $\psi_k$ under condition \eqref{LL2}.

\subsection{Asymptotics of  eigenfunctions}

 Let us find an asymptotic  for $\psi_k$ and $w_k$  with respect to the parameter
$k/b$, which is small due to Lemma \ref{L2N}. Using the Taylor expansion for $\psi_k$ at $y=0$
we introduce the function $\bar W_k$ as a solution of the equation
\begin{equation} \label {BV1a}
L  \bar W_k - \lambda_k \bar W_k= C_U r b^4 \rho_2 y^2 \exp(-by), 
\end{equation}
where, according to Lemma \ref{Ldps},  without any loss of generality one can set $\rho_2=1$.
Solving eq. (\ref{BV1a}) one has  
\begin{equation} \label {mainB1}
\bar W_k=C_U r b^4 D_b^2 \Big (\frac{\exp(-by) - \xi_k \exp(-\bar k y)}{b^2 -\bar k^2}\Big ),
\end{equation}
where 
$$
\xi_k=\frac{1+r}{r(1+ \bar k/\beta)},  \quad D_b=\frac{\partial }{\partial b}.
$$
We represent  $\psi_k$ by 
\begin{equation} \label{mainB2}
\psi_k = \bar \Psi_k  + \tilde \Psi_k,
\end{equation}
where
$$
\bar \Psi_k=C_U r k^2 b^4   D_b^2 \Big((b^2 -\bar k^2)^{-1} (\Phi_{k}(b,y) - \xi_k \Phi_{k}(\bar k, y))\Big)
$$
and
\begin{equation} \label {PhB}
 \Phi_{k}(p, y) =\frac{\exp(- p y) -  \exp(- ky) + y(p -k) \exp(-ky)}{(p^2 -k^2)(p^2 -k^2 -\lambda \nu^{-1})}. 
\end{equation}
We see that 
\begin{equation} \label{mainterm}
\bar \Psi_k(\bar k, y) = -6C_U k^2   \tilde \xi_k \Phi_{k}(\bar k, y) + Z_k(\bar k, y), 
\end{equation}
where 
\begin{equation} \label{txik}
  \tilde \xi_k=\frac{1+r}{1+\bar k\beta^{-1}},
\end{equation}
and $Z_k$ is defined by relations
\begin{equation} \label{corrterm}
Z_k(\bar k, y)=Z_k^{(1)}(\bar k, y) + Z_k^{(2)}(\bar k, y).
\end{equation}
Here
$$
Z_k^{(1)}(\bar k, y) =C_U r k^2 b^4 D_b^2 (\frac{\Phi_k(b,y)}{b^2 -\bar k^2}),
$$
$$
Z_k^{(2)}(\bar k, y) =-C_U k^2   \tilde \xi_k \big(b^4 D_b^2 ((b^2 -\bar k^2)^{-1}) -6 \big ) \Phi_k(\bar k,y).
$$
By these relations and Lemma  \ref{L2N}  one obtains 
\begin{equation} \label{esttildeZ}
\sup_{y \in [0,h]} |Z_k(\bar k, y)| < c_6 |C_U| |\bar k^2| /b^2
\end{equation}
and thus the function $Z_k$ is a small correction to the main term $ \bar \Psi_k$.

\begin{lemma} \label{asmainterm} {In the domain $ \lambda \in {\mathbb C}_{1/2}$ 
the function $\bar \Psi_k(\bar k, y) $ is an analytic in $\lambda$ and for sufficiently large $b>0$ satisfies  the estimate
\begin{equation} \label{mainterm2}
|\bar \Psi_k(\bar k, y) | <  C_0 | C_U|,  
\end{equation}
where $C_0>0$ is a constant independent of $k$ and $b$.
}
\end{lemma}

\begin{proof}   

Note that according to definition \eqref{rbb} of $\beta$ and $r$,  one has $|\tilde \xi_k| < 2$ for $ \lambda \in {\mathbb C}_{1/2}$. 
Consider the function $\Phi_k(\bar k, y)$.  We make the substitution $\lambda=k^2 \tau$, where $\tau$ is a new complex variable defined 
in the domain $Re \ \tau > -1/2$. Then $\bar k -k =k(\sqrt{1+\tau}-1)$.  Using  the Taylor series one has 
\begin{equation} \label{Phtau}
\Phi_k(\bar k, y)= \frac{y^2 (\sqrt{1+\tau}-1)^2}{2 k^2 \tau^{2}(1- \nu^{-1}) } T(ky, \tau)\exp(-ky), 
\end{equation}
where
\begin{equation} \label{Phtau1}
T(z, \tau)=1 - 2z \frac{\sqrt{1+\tau}-1}{3!} + 2z^2 \frac{(\sqrt{1+\tau}-1)^2}{4!} - ...    
\end{equation}
Note that the function $ (\sqrt{1+\tau}-1)^2 \tau^{-2}$ is uniformly bounded for $Re \ \tau > -1/2$.  Since $(ky)^m \exp(-ky) \le  m^m \exp(-m)  $ the series
for $T$ converges and uniformly bounded  in the domain $|\sqrt{1+\tau}-1| <1$. Consequently, in that domain
one has
\begin{equation} \label{Phtau1}
|\Phi_k(\bar k(\tau), y)|  <   C_0. 
\end{equation}
 On the other hand for $|\tau| > \delta$ we have 
\begin{equation} \label{Phtau2}
|\Phi_k(\bar k, y)|= \frac{| \exp(- \bar k y) -\exp(- ky) + y (\bar k -k) \exp(-ky)|} {k^{4} |\tau^{2}| (1- \nu^{-1})}<  c_0(\delta) k^{-2}. 
\end{equation}
Using \eqref{Phtau1}, \eqref{Phtau2}, \eqref{PhB} and \eqref{mainterm} one obtains \eqref{mainterm2}.
\end{proof}

\subsection{Perturbation theory}

We represent $\psi_k$ by   \eqref{mainB2}. Then 
 for  $\tilde \Psi_k$  one obtains the equation
\begin{equation} \label {BI1}
\lambda \nu^{-1} L_k\tilde \Psi_k = L_k^2 \tilde \Psi_k + k^2  \tilde  W_k,  
\end{equation}
where $\tilde W_k$ satisfies
\begin{equation} \label {BI2}
  \lambda_k  \tilde  W_k =  L_k  \tilde  W_k +  \mu y P_N(y)  (\bar \Psi_{k} + 
 \tilde \Psi_{k})  + S_{k},   
\end{equation}
and $S_{k}$ admits the estimate
\begin{equation} \label{BI2E}
  |S_{k}(y) | < C_3  r b^4  y^{3} \exp(-by) |D_y^{3} ( \bar \Psi_{k}+ \tilde \Psi_k  )|_{\infty}. 
\end{equation}
The functions $\tilde \Psi_k$ and $\tilde W_k$ satisfy the boundary conditions 
\begin{equation} \label {Ebw}
   \frac{d\tilde W_k(y)}{dy}\vert_{y=0}=\beta \tilde W_k(0), \quad  \frac{d\tilde W_k(y)}{dy}\vert_{y=h}=\beta_1 \tilde W_k(h), 
\end{equation}
\begin{equation} \label {Ebp1}
  \tilde \Psi_k(0)=0, \quad  \frac{d\tilde \Psi_k(y)}{dy}\vert_{y=0}=0, 
\end{equation}
\begin{equation} \label {Ebp2}
  \tilde \Psi_k(h)=p_k, \quad  \frac{d\tilde \Psi_k(y)}{dy}\vert_{y=h}=q_k, 
\end{equation}
where $|p_k|, |q_k| < c_1\exp(- k h) < c_2 b^{-10}$.

We can resolve  the BVP  defined by (\ref{BI1}), (\ref{BI2}), (\ref{Ebw}),  (\ref{Ebp1}) and (\ref{Ebp2})  by iterations that follows from the next lemma.  That
BVP defines $\tilde \Psi_k$ via $\bar \Psi_k$,  i.e., $\tilde \Psi_k=A(\bar \Psi_k)$, where $A$ is a linear operator. We consider that operator  on the space $C^3[0,h]$ of functions $f$
with the bounded norm 
$$|f|_3=\sum_{m=0}^3 \sup_{y \in [0,h]} |D_y^m f|.$$

\begin{lemma} \label{Pert} {For   $|\bar k| < c_0 b^{s_1}$, where $s_1 \in (0,1)$,  and sufficiently large $b$ the operator $A$ is a contraction on $C^3[0,h]$. More precisely,  the solutions of the BVP  defined by (\ref{BI1}), (\ref{BI2}), (\ref{Ebw}),  (\ref{Ebp1}) and (\ref{Ebp2})  satisfy 
\begin{equation} \label{estBI}
| \tilde \Psi_k|_{3}   <  c_1 b^{-s_4} | \bar \Psi_k|_{3} , \quad s_4=\min\{s_3,  4s_0/5 \}>0.
\end{equation}}
\end{lemma}
\begin{proof} 
We have $\tilde W_k(y) =J_P(y) +J_S(y)$, where
\begin{equation} \label{JP}
J_P=\mu \int_0^h   \Gamma_{\bar k} (y, y_0)  y_0 P_N(y_0)  (\bar \Psi_k(y_0) +  \tilde \Psi_k(y_0))dy_0,
\end{equation}
\begin{equation} \label{JS}
J_S= \int_0^h   \Gamma_{\bar k} (y, y_0)  S_k(y_0) dy_0. 
\end{equation}
By \eqref{Lnu} and using that
for large $b$ one has $\mu \max_{y \in [0,h]} |y P_N(y)| < c_4 b^{-s_3}$, one finds
\begin{equation} \label{JP1}
|J_P| \le   c_2 b^{-s_3}  |\bar k|^{-1} k^{-1} (|\bar \Psi_k|_{\infty}  +  |\tilde \Psi_k|_{\infty}).
\end{equation}
To estimate $|J_S|$ we note that 
$$
\int_0^h  | \Gamma_{\bar k} (y, y_0)| y_0^3 \exp(-by_0) dy_0 \le c_3 b^{-4} |\bar k|^{-1}.
$$
Thus, by \eqref{BI2E} one has
\begin{equation} \label{JS1}
|J_S|  <  c_4 r |\bar k|^{-1} (|D_y^3 \bar \Psi_k|_{\infty}  +  |D_y^3 \tilde \Psi_k|_{\infty}).
\end{equation}
Consequently
\begin{equation} \label{Wkk}
|\tilde W_k|_{\infty}  <  c_5\Big( b^{-s_3} \frac{1}{ |\bar k| k}  (|\bar \Psi_k|_{\infty}  +  |\tilde \Psi_k|_{\infty}) + \frac{ r }{|\bar k|} (|D_y^3 \bar \Psi_k|_{\infty}  +  |D_y^3 \tilde \Psi_k|_{\infty})\Big).
\end{equation}
From Lemma \ref{int} it follows that 
\begin{equation} \label{kaka}
|\tilde \Psi_k|_{\infty} \le  c_6 k^{-2} h |\tilde W_k|_{\infty}, \quad |D_y^3\tilde \Psi_k|_{\infty} \le  c_7 k h |\tilde W_k|_{\infty}.
\end{equation}
Substituting those inequalities into \eqref{Wkk} and noticing that for large $b$ one has $r h < b^{-4s_0/5}$, we find that
\begin{equation} \label{Wkk2}
|\tilde W_k|_{\infty}  <  c_7\Big( b^{-s_3} ( |\bar k|k)^{-1}  |\bar \Psi_k|_{\infty}  + r |\bar k|^{-1} |D_y^3 \bar \Psi_k|_{\infty}\Big).
\end{equation}
Again using \eqref{kaka}, one has \eqref{estBI}.\end{proof}

\subsection{Nonlinear equation for $\lambda_k$}

Let us make the substitution $z=\bar k(\lambda)/k$, where $z$ is a new complex unknown.  Since $\lambda \in {\mathbb C}$, for each fixed $k$ the variable $z$
lies in the domain 
$$
{\mathbb D}_{k,b}=\{z \in {\mathbb C}:  Re \ z > \sqrt{k^2 -1/2},  \    |z| <c_0 h rb/k \}.
$$
According to \eqref{PhB}, \eqref{mainterm},  \eqref{corrterm}, \eqref{Phtau} and \eqref{Phtau1}   one has
\begin{equation} \label{d22psi}
\frac{d^2\bar \Psi_k(y, \lambda)}{dy^2}\vert_{y=0}=-3 \tilde \xi_k(z)  C_U (1- \nu^{-1})^{-1} g(z) + H_k( z, b),
\end{equation}
where 
$$
g(z)=\frac{1}{(z +1)^2},   \quad \tilde \xi_k(z)=\frac{1+r} {1 + kz/rb},  \ 
$$
$$
H_k=H_{k,0}  + H_{k,1},
$$
$$
H_{k,0}=C_U r b^4 D_b^2\Big(   \frac{b-k} {(b^2 - k^2 z^2)(b+k)(b^2 - k^2 (1 + (z^2-1) \nu^{-1}))}        \Big),
$$
$$
2H_{k,1}=-C_U \tilde \xi_k(z) \big(b^4 (D_b^2(b^2 - k^2z^2)^{-1})   - 6      \big) g(z).
$$
Let us set 
\begin{equation} \label{s2s0}
    3C_U =-8(1- \nu^{-1})(1+r)^{-1}.
\end{equation}
We note that $|1+z| > 1$ and thus according to  estimate \eqref{esttildeZ} and Lemma \ref{L2N} for large $b$  one has  
\begin{equation} \label{d2G}
\sup_{z \in {\mathbb D}_{k,b}}| H_k(z, b)| <  c_1 b^{-2s_0}.
\end{equation}

Let us introduce
$$
\tilde H_k(z, b) :=\frac{d^2 \tilde \Psi_k(y, \lambda(z))}{dy^2}\vert_{y=0}.
$$
Due to Lemma \ref{Pert} the term  $\tilde H_k$  is a smooth uniformly bounded function in the domain ${\mathbb D}_{k,b}$:
\begin{equation} \label{HH1}
 \sup_{z \in {\mathbb D}_{k,b}}|\tilde H_k(\tau, b)| <  c_3 b^{-s_4}. 
\end{equation}

Then we use that $\psi_k=\bar \Psi_k + \tilde \Psi_k$ and compute  $\rho_2$ defined by \eqref{LL2p}. Without any loss of generality one can  set $\rho_2=1$.
 As a result,   we obtain the equation for $z$:
\begin{equation} \label{Spmu1}
 (z+1)^2  = 4  (1 + a z)^{-1}  +  Y_k(z, b),
\end{equation}
where
$$
2Y_k=  (1+z)^2( \tilde H_k(z, b) + H_k(z, b)) ,  \quad  a=k/rb.
$$
We note  that $Y_k$ admits the estimate
\begin{equation} \label{YYY}
 |Y_k( z, b) | <  C_1 b^{-s_4} |1+z|^2,  \quad z \in {\mathbb D}_{k,b}.
\end{equation}

\subsection{Investigation of equation \eqref{Spmu1}}

As $b \to \infty$ we have $Y_k \to 0$ and the limit cubic equation, which arises from  \eqref{Spmu1},  has three roots: a real  and two imaginary ones.

\begin{lemma}  \label{YL} 
{
({\bf a})  If
the root $z_k$ of \eqref{Spmu1} lies in the domain 
$$
{\mathbb E}_{k,b, c_1}=\{ z \in {\mathbb D}_{k,b}: \  Re \ z  > 1 - c_1 b^{-s_4} \},
$$ 
then $z_k \in I_b$, where
\begin{equation} \label{zzz2}
  I(b)= \{ z \in  {\mathbb C} :    |z-1|  <  c_4 b^{-s_4} \},
\end{equation}
where $c_4 >0$ is a constant. The subdomain $I_b$ contains only a single root of \eqref{Spmu1};

({\bf b}) 
Let
\begin{equation} \label{yyy}
 Re \ Y_k(z,b) <  -   c_1 b^{-s_4}, \quad  z \in I(b).
\end{equation}
Then for sufficiently large $b$ the root $z_k$ of \eqref{Spmu1} satisfies
\begin{equation} \label{zzz}
 Re \ z_k<  1 -   c_2 b^{-s_4}.
\end{equation}}

\end{lemma} 

\begin{proof}
({\bf a})
Since $Re \ z_k >0$,  we have $|1+ az_k| >1$ and then \eqref{Spmu1} implies that
\begin{equation} \label{1+z}
 |1 +z_k| < 2 + c_1 b^{-s_4} < 3.
\end{equation}
Then \eqref{YYY} can rewritten as 
\begin{equation} \label{YYY1}
 |Y_k( z, b) | <  C_2 b^{-s_4},  \quad z \in {\mathbb D}_{k,b}.
\end{equation}
We take the imaginary part of \eqref{Spmu1} and one has
$$
2 (Re \ z_k +1 ) Im \ z_k =- \frac{4a  Im \ z_k }{1+ a^2 |z_k|^2}  + Im \ Y_k.
$$
That relation and \eqref{YYY1} give 
\begin{equation} \label{Imz}
 |Im \ z_k| <  c_2 b^{-s_4}.
\end{equation}
Now 
we take the real part of \eqref{Spmu1} that entails
\begin{equation} \label{Rez0}
 (Re \ z_k +1)^2= (Im \ z_k )^2 +  4(1 - a  Re \ z_k)(1+ a^2 |z_k|^2)^{-1}  + Re \ Y_k.
\end{equation}
By that relation, estimates \eqref{YYY1},  \eqref{Imz}, $Re \ z_k > 0$ and $a >0$ one obtains that
\begin{equation} \label{Rez1}
 Re \ z_k    <   1+ c_2 b^{-s_4}. 
\end{equation}
That inequality and \eqref{Imz} entail the assertion {(\bf a)}.

Let us prove the uniqueness of the roots $z_k$ in the case ({\bf a}).  
One has  $z_k(a) \in I(b)$.   For bounded $k$ the distance $d_{E, D}=dist(I_b, \partial {\mathbb D}_{k, b})$ between the boundary of the domain ${\mathbb D}_{k, b}$ and $I_b$ satisfies $d_{E, D} > 1/4$.  Note that $H_k(z, b)$ is an analytical  function of $z$ in the domain ${\mathbb D}_{k, b}$ (see the proof of Lemma  \ref{asmainterm}). Therefore, estimate \eqref{d2G} entails
\begin{equation} \label{derH}
|\frac{dH_k}{dz}|  <  c_3  b^{-2s_0},   \quad  z \in I_b.
\end{equation}
The perturbation $\tilde H_k$ also is analytic in $z$ in  the domain ${\mathbb D}_{k, b}$. Indeed,  $\tilde H_k$ is the derivative of $\tilde \Psi_k$, which is 
a fixed point of of contraction mapping and that contraction analytically depends on the  parameter $z$ if $z \in  {\mathbb D}_{k, b}$.  Therefore, 
\eqref{HH1} gives
\begin{equation} \label{derH1}
|\frac{d\tilde H_k}{dz}|  <  c_7  b^{-s_4},   \quad  z \in I_b.
\end{equation}
Estimates \eqref{derH} and \eqref{derH1} imply that
\begin{equation} \label{derH2}
|\frac{dY_k}{dz}|  <  c_8 b^{-s_5},   \quad    z \in I(b), \  s_5=\min\{s_4, 2s_0 \}
\end{equation}
for some $c_8 >0$.
Now the uniqueness of the root $z_k=0$ for $k=k_1,..., k_N$ follows from \eqref{derH2} and the Implicit Function Theorem.

Consider assertion ({\bf b}).   If $Re \ Y_k < -c_1 b^{-s_4}$, then relations \eqref{Imz} and  \eqref{Rez0} imply
\begin{equation} \label{Rez1}
 Re \ z_k +1 <  -c_2  b^{-s_4} +   4 (1+ a^2 |z_k|^2)^{-1},
\end{equation}
and \eqref{zzz} follows.
\end{proof}

For large $b$ and  $z \in I(b)$ the main term of the asymptotics of $\tilde H_k$ 
can be found by the system of equations 
\begin{equation} \label {PT1}
L_k^2  \Psi_k^{(0)}=k^2 W_k^{(0)},  
\end{equation}
\begin{equation} \label {PT2}
 L_k  W_k^{(0)} =    y^3 P_N(y) \exp(-ky) :=V_k(y),     
\end{equation}
for the unknown functions $\Psi_k^{(0)}(y)$ and $ W_k^{(0)}(y)$, which satisfy no-flip  boundary conditions analogous to ones for   $\tilde \Psi_k$ and $ \tilde W_k$.
For $z=1$ the perturbation $\tilde H_k(z, b)$ can be expressed
via $\Psi_k^{(0)}(y)$ by the relation
\begin{equation} \label {PTH}
\tilde H_k(1, b)= \mu \frac{d^2\Psi_k^{(0)}(y)}{dy^2}\vert_{y=0}  + O(\mu b^{-s_4}).
\end{equation}

By a variation of   $P_N(y)$ in \eqref{Uy}, we can obtain  a large class of  perturbations  $\tilde H_k$   that
follows from the next lemma.

\begin{lemma} \label{pertcon}
{
For any polynomial $Z_N(p)= q_0+ q_1 p+ ... + q_N p^N  $ of variable $p=k^{-1}$
 there exists a polynomial $\tilde P_N(y)=\sum_{n=0}^N \bar r_n y^n$ such that
\begin{equation} \label {PT3}
 D_y^2  \Psi_k^{(0)}(y) \vert_{y=0}= k^{-6}\big( Z_N(k^{-1})  +  \frac{k}{(\beta + k)} \tilde Z_N(k^{-1}) +  O(b^{-10}) \big), \quad b \to \infty
\end{equation}
 where
\begin{equation} \label {PT4}
  \tilde Z_N(k^{-1})=a_0 q_0+ a_1 q_1 k^{-1} + ... + a_N q_N k^{-N}, 
\end{equation}
and $a_n >0,  n=1,..., N$ are some coefficients independent of $k,b$.}
\end{lemma}

\begin{proof}  
Consider equation \eqref{PT1}
under no-flip conditions at $y=0, h$. We multiply the both hand sides of \eqref{PT1}  by $f_k(y)=(ky +1)\exp(-ky)$.  Integrating by parts on $[0, h]$ and using that
$df_k(y)/dy=0,  \Psi_k^{(0)}(y)=0$ and 
$d\Psi_k^{(0)}/dy=0
$ at $y=0$, we obtain
\begin{equation} \label {PT2b}
 \frac{d^2  \Psi_k^{(0)}}{dy^2}\vert_{y=0} =   \int_0^h  W_k^{(0)}(y) (ky+1) \exp(-ky)  dy  +  O(\exp(-kh)).   
\end{equation}
Let 
$$
\rho_k(y)=-(\frac{3}{4k(\beta +k)} + \frac{3y}{4k} +\frac{ y^2}{4})\exp(-ky).
$$
Note that $L_k \rho_k =f_k(y)$.
Suppose  $w(y)$ be a smooth function defined on $[0,h]$ and satisfying the condition
$$
\frac{dw(y)}{dy}\vert_{y=0}=\beta w(0).
$$
The function $\rho_k$ also satisfies the same boundary condition.
Then, integrating by parts one obtains
$$
 \int_0^h  (L_k w(y))  \rho_k(y) dy=O(\exp(-kh)) + \int_0^h  w(y)  (ky +1) \exp(-ky) dy. 
$$
We multiply the both hand sides of eq. \eqref{PT2} by $\rho_k$. Then, again  integrating by parts, and using the last equality, we find
that
\begin{equation} \label {PT2ab}
 \frac{d^2 \Psi_k^{(0)}}{dy^2}\vert_{y=0} =   \int_0^h y^3 P_N(y) \exp(-ky)  \rho_k(y)\ dy  +  O(b^{-10}).   
\end{equation}
We substitution $P_N$ into that relation that gives 
\begin{equation} \label {PT2abm}
 \frac{d^2 \Psi_k^{(0)}}{dy^2}\vert_{y=0} =   \sum_{n=0}^N  \bar r_n (2k)^{-n-6}\Big(\frac{3k(n+3)! }{\beta +k} +  
\frac{3(n+4)! }{2}+ \frac{(n+5)! }{4}\Big)
+  O(b^{-10}).   
\end{equation}
Then using  relation \eqref{PT2abm} we can find coefficients $\bar r_n$ such that
$$
k^{-6} \sum_{l=0}^N  q_l k^{-l} =  \sum_{n=0}^N \bar r_n (2k)^{-n-6}\Big(  
\frac{3(n+4)! }{2}+ \frac{(n+5)! }{4}\Big)
$$
for all $k \ne 0$. As a result, we obtain \eqref{PT3} and \eqref{PT4}, where
$$
a_n=3\Big(\frac{3(n+4)}{2} +  \frac{(n+4)(n+5)}{4} \Big)^{-1},
$$
 that  entails the assertion of the Lemma.
\end{proof}

In order to use the last lemma  we
 take $Z$ of a special form:
\begin{equation} \label{SQpr}
Z(p,d)=- \Big(\prod_{j=1}^N (p - (k_j+ d_j)^{-1} \Big)^2,
\end{equation}
where $k_j$ are defined by Prop. \ref{6.2} and $d_j$ are parameters. The polynomial $\tilde Z_k$ is defined then by 
\eqref{PT4}.

\begin{lemma} \label{mainL}
{For sufficiently large $b >0$ one can find  coefficients $d_j(b)$ such that
the real roots $z_{k}$ of eq. (\ref{Spmu1}) satisfy
\begin{equation} \label{A1}
z_{k} < 1 - c_0 b^{-c_1},              \quad  \ k \notin { \mathcal K}_N,
\end{equation}
for some $c_1 >0$ independent of $b$ and
\begin{equation} \label{A2}
z_{k}=1,              \quad  \ k \in { \mathcal K}_N.
\end{equation}
Moreover,  $d_j(b) \to 0$ as $b \to +\infty$.}
\end{lemma}

\begin{proof}
Let $d_j$ be small enough, for example, $|d_j| < 1/10$ for all $j=1,..., N$.
Let us introduce a dependence on $d$ in notation  for $Y_k$. We use Lemmas  \ref{YL} and \ref{pertcon} noticing that for $z \in I(b)$
\begin{equation} \label{YY55}
Y_k( z, b,d)=  \mu k^{-6} (Z(k^{-1}, d) + \frac{k}{(\beta+k)} \tilde Z(k^{-1}, d)) + O(\mu b^{-s_{4}})  + O(b^{-2s_0}).
\end{equation}
Now we make a special choice of the exponents $s_0$ and $s_2$.  Let us take $s_2 \in (0, 1/10)$. Then we choose $s_0 \in (1- s_2/2,1)$ close to $1$ such that
\begin{equation} \label{s0s2C}
  s_*=(6+ 2N)(1-s_0) >    s_4.
\end{equation}
Using Lemma \ref{L2N} and definitions \eqref{SQpr} and \eqref{PT4} of $Z$ and $\tilde Z$,    we see then that for all positive $k$ such that $k < \min \{ \beta^{1/2},
\ c_1 b^{1-s_0} \}$, 
and  $k  \notin { \mathcal K}_N $   the term
$\mu k^{-6} Z(k^{-1}, d)$ is dominant in the right hand side of \eqref{YY55} and that term is negative. One has $Z(k^{-1}, d) <  -c_0 b^{-s_*} < 0$. 

For all positive $k$ such that $k >\beta^{1/2}$ but $k < 
c_{9} b^{1-s_0}\}$,    the both terms
$\mu k^{-6} Z(k^{-1}, d)$ and $\mu k^{-6} \frac{k}{(\beta +k)}\tilde  Z(k^{-1}, d)$ are dominant (as $b \to +\infty$) in the right hand side of \eqref{YY55}  with respect to the contribution
$O(b^{-s_4})$.  Those terms are both   negative. Thus we obtain that 
for $k  \notin { \mathcal K}_N $  the  roots of eq. \eqref{Spmu1} satisfy \eqref{A1}.

Consider \eqref{A2}.
We define $d_j$ by the system
\begin{equation} \label{Spmu1K}
  \frac{rb+k_j}{4rb(1+ r)} +  \mu Y_{k_j}(z, b, d)\vert_{z=1}=0, \quad \forall \ j=1,..., N.
\end{equation}
Let us prove that this system has a solution $d$ such that $|d|$ is small enough. We use relation \eqref{YY55}.
Since $
Z(k_j^{-1}, 0) =0$ and the Jacobian matrix defined by  
$$
J_{jl}(d)=\frac{\partial Z(k_j^{-1}, d)}{\partial   d_l}
$$ is invertible for small $|d|$ (it follows from definition  \eqref{SQpr} of $Z$)
  we can apply to eq. \eqref{Spmu1K} the Implicit Function Theorem and find
a solution $d$ of \eqref{Spmu1K}  such that $|d| \to 0$ as $b \to 0$. Now  relation \eqref{A2}
follows from Lemma \ref{YL}.
\end{proof}

\begin{proof}  of Prop. \ref{6.2}.
Proposition \ref{6.2} follows from the last lemma.
\end{proof}


\subsection{Conjugate spectral problem }
\label{sec8.5}

Standard computations show that the  corresponding conjugate spectral problem  is defined by the following equations and the boundary conditions:
\begin{equation}
\lambda \Delta \phi  =\nu \Delta^2  \phi  - U_y \tilde w_x,
\label{conj4S}
\end{equation}
\begin{equation}
\lambda \tilde w = \Delta  \tilde w   + \nu \phi_x, 
\label{conj3S}
\end{equation}
\begin{equation}
  \tilde w_x(x, y)\vert_{x=0, \pi} =0,
\label{boundNeumS}
\end{equation}
\begin{equation}
  \tilde w_y(x, y)\vert_{y=0} =\beta \tilde w(x, 0), \quad  \tilde w_y(x, y)\vert_{y=h} =\beta_1 \tilde w(x, h),
\label{boundDirS2}
\end{equation} 
\begin{equation}
 {\phi} (x,y)\vert_{y=0,h}   = \phi_y(x, y)\vert_{y=0, h}=0,
\label{boundstream2}
\end{equation}
\begin{equation}
 {\phi} (x,y)\vert_{x=0, \pi}   = \Delta \phi(x, y)\vert_{x=0,\pi}=0.
\label{boundstream2x}
\end{equation}
The Fourier decomposition of that  problem   is  as follows:
\begin{equation}
\lambda_k L_k \phi_k=\nu L_k^2 \phi_k  - k^2 U_y \tilde w_k,
\label{conj4}
\end{equation}
\begin{equation}
\lambda_k \tilde w_k=L_k \tilde w_k   + \nu  \phi_k,
\label{conj3}
\end{equation}
\begin{equation}
\frac{d\phi_k}{dy}(y)\vert_{y=0,h}= \phi_k(0)=\phi_k(h)=0,
\label{conj5}
\end{equation}
\begin{equation}
          \frac{d\tilde w_k}{dy}(y)\vert_{y=0}=\beta \tilde w_k(0), \quad  \frac{d\tilde w_k}{dy}(y)\vert_{y=h}=\beta_1 \tilde w_k(h).
\label{conj6}
\end{equation}

\subsection{ Eigenfunctions of $L$ and $L^*$ with zero eigenvalues}
\label{sec8.4}

Let us consider the eigenfunctions $e_k$ of $L$ with the zero eigenvalues. By the stream function,  they can be represented as     have the form
\begin{equation}
e_j=({\bf v}_j, \theta_j)^{tr},        
\label{eig1}
\end{equation}
where $j=1, ..., N$,  ${\bf v}_j=(D_y \psi_j,  -D_x \psi_j)^{tr}$ and
$$
\psi_j=\Psi(y) \sin(k_j x), \theta_j(y) =\Theta_j\cos(k_j x).
$$

Relations   \eqref{PhB} and \eqref{mainterm}  show that
\begin{equation}
\Psi_j(y)= y^2 \exp(-k_j y) + \tilde \psi_j ,    
\label{eig1a}
\end{equation}
where $||\tilde \psi_j|| \le c_1 b^{-c_2}$ for some $c_i >0$. 
Moreover, by \eqref{mainB1} one obtains
\begin{equation}
\Theta_j(y)=  \beta_j\big( \exp(-by) -  \exp(-k_jy) +    \tilde \theta_j(y) \big),  
\label{eig2}
\end{equation}
where $||\tilde \theta_j  ||< c_0  b^{-c_1}$ for some $c_i >0$,  and $\beta_j$ are constants.

The eigenfunctions $e_j=({\bf v}_j^*, \theta_j^*)^{tr}$ of $L^*$  with the zero eigenvalues satisfy similar relations.  To obtain them, we
consider conjugate spectral problem (\ref{conj4}), (\ref{conj3}).
Asymptotics of solutions 
of that system for large $b$ can be found as above. Without any loss of generality, we can set $\tilde w_k(0)=1$ in 
eq. (\ref{conj4}). Then,  to obtain a main term of asymptotics for $\phi_k$ as $b \to \infty$, we replace $U_y \tilde w_k$ by $C_U r b^4 \exp(-by)$ in that equation.  Then for bounded $k$
a simple computation gives
\begin{equation}
 \phi_k=\nu^{-1} r  C_U \Big(\exp(-by)- \exp(-ky)  + y(b -k) \exp(-ky)  + \hat \phi_k \Big)
\label{eig1mphi}
\end{equation}
where $|\hat \phi_k  | < c_3 b^{-c_4}$, $c_i >0$.
Substituting that relation into \eqref{conj3},
after some computations we obtain 
\begin{equation}
   \theta_j^*=\Theta_j^* (y) \cos(k_j x),
\label{eig1m}
\end{equation}
where $j=1, ..., N$ and
\begin{equation}
\Theta_j^*= \tilde a_j \Big((k_j y^2  + y)\exp(-k_j y) + \tilde \theta^*_j \Big),
\label{eig3m}
\end{equation}
where $||\tilde \theta^*_j  ||< c_4 b^{-c_5}$, $c_i >0$, $\tilde a_j$ are constants, which should be chosen to satisfy biorthogonality relations
$$
\langle e_j, e_i^* \rangle =\delta_{i, j},
$$
where $\delta_{i,j}$ stands for the Kronecker symbol.    For $v_j^*$ one has 
$v_j=(D_y \psi^*_j , - D_x \psi^*_j )^{tr}$.
As it follows from \eqref{eig1mphi}  the norms of the functions $\psi^*_j$ are small: 
\begin{equation}   \label{estz}
||\psi^*_j|| < C\nu^{-1}.
\end{equation}

The next lemma 
concludes the investigation of spectral properties of the operator $L$.

\begin{lemma} \label{simplespec}
{The eigenvalue $0$ of the operator $L$ has no generalized eigenfunctions.}
\end{lemma}

\begin{proof} Let us check that generalized eigenfunctions are absent.
Since $0$ has a finite multiplicity, we can use the Jordan representation.  Assume that there exists a generalized eigenfunction  $e_g$.
Then there is a non-zero $b$ such that
$
L e_g=b=\sum_{j=1}^N b_l e_l
$
for some $b_l$, $ l \in \{1, ..., N\}$.  Then  $\langle b,  e_k^* \rangle=0$ for all $k \in \{1,..., N\}$. Eigenfunctions $e_k^*$ and $e_l$ are orthogonal
for $k \ne l$, and thus  all coefficients $b_l=0$. \end{proof}

\subsection{ Estimates for semigroup $\exp(Lt)$ }
\label{sec: 8.6}

The operator $L$ is sectorial and, according to  Prop. \ref{6.2},   satisfies  the Spectral Gap Condition. Therefore \cite{He}
for some $\kappa > 0$ and for all $v=(\tilde \omega, w)^{tr}$ such that
$$
\langle v, \tilde e_j \rangle =0, \quad j=1, ..., N,
$$
one has the following estimates 
\begin{equation}
|| \exp(Lt) v|| \le M \exp(-  \rho t) ||v||,  \quad
\label {semigroup1}
\end{equation}
\begin{equation}
|| \exp(Lt) v||_{\alpha} \le \bar M S_{\alpha}(t) \exp(- \rho t) ||v||
\label {semigroup2}
\end{equation}
Here $\alpha > 1/2$ and 
$$
S_{\alpha}(t) = t^{-\alpha}, \quad 0 < t \le \rho^{-1},
$$
and
$$
S_{\alpha}(t) = \rho^{-\alpha}, \quad  t > \rho^{-1}.
$$
The constants $M, \bar M$ and $\rho$ can depend on $b$ but they are independent of $\gamma$.
Estimates (\ref{semigroup1}) and (\ref{semigroup2}) will be used in the proof of  existence of the invariant manifold, see Appendix.

\section{Finite dimensional invariant manifold}
\label{sec: 9}

 In this section, we reduce the Navier -Stokes  dynamics to a system of ordinary differential equations 
  following  Sect. 4.  Let  $E_N$ be
the finite dimensional subspace
$E_N=Span \{e_1,  ..., e_{N} \}$
of the phase space ${\mathcal H}$, where $e_j=({\bf v}_j, \theta_j)^{tr}$ are the  eigenfunctions of the operator
$L$ with the  zero eigenvalues. Let $e_j^*=({\bf v}_j^*, \theta_j^*)^{tr}$ be the corresponding eigenfunctions of the conjugate operator
$L^*$.  We assume that $e_j^*$ and $e_j$  are biorthogonal.
Let ${\bf P}_N, {\bf Q}_N$ be  the  projection operators
\begin{equation}
{\bf P}_{N} z= \sum_{j=1}^{N} \langle z, e_j^* \rangle e_j, \quad {\bf Q}_N={\bf I} - {\bf P}_N.
\label{Pr2}
\end{equation}
We denote by ${\bf P}_{v, N}, \ {\bf P}_{w, N},  {\bf Q}_{v, N}$ and  ${\bf Q}_{w, N}$  the components of these operators. For  ${\bf P}_{v, N}$ and ${\bf P}_{w, N}$ one 
has
\begin{equation}
{\bf P}_{v, N} z= \sum_{j=1}^{N} \langle z, e_j^* \rangle {\bf v}_j, \quad {\bf P}_{w,N}=\sum_{j=1}^{N} \langle z, e_j^* \rangle { \theta}_j.
\label{Pr2v}
\end{equation}

We transform equations (\ref{eveq10}),(\ref{eveq11}) into a  system with "fast" and "slow" variables.
 Let us introduce the auxiliary functions  $R_{v}(X)$ and $R_{w}(X)$ by
$$
 R_{v}(X) =\sum_{j=1}^{N} X_j {\bf v}_j,  \quad R_w(X)= \sum_{j=1}^{N} X_j \theta_j,
$$
where $X=(X_1, ..., X_N)^{tr}$.
We represent $ \tilde {\bf v}$ and $w$ by
\begin{equation}
\tilde {\bf v}= R_{v}(X) + \hat {\bf v},  \quad w=  R_w(X) +\hat w,  
\label{om2}
\end{equation}
where
$$
{\bf P}_N (\hat {\bf v}, \hat w)^{tr}=0,
$$
and $\hat w, \hat {\bf v}$ and $X(t)$ are new unknown functions.

We  substitute relations (\ref{om2}) in eqs. (\ref{eveq10}) and (\ref{eveq11}).
 As a result,  one obtains  the system
\begin{equation}
\frac{dX_i}{dt}= \gamma   F_i( X, \hat {\bf v},\hat w),
\label{X2}
\end{equation}
\begin{equation}
\hat {\bf v}_t={\mathbb P}\Big(\nu \Delta \hat {\bf v} + \kappa {\bf e}\hat w+ \gamma {\bf Q}_{1,N} F(X, \hat {\bf v},\hat w) \Big),
\label{hom2}
\end{equation}
\begin{equation}
\hat w_t= \Delta \hat w  - \hat v_2 U_y +  \gamma {\bf Q}_{2,N} G( X, \hat {\bf v}, \hat w),
\label{hw2}
\end{equation}
where
\begin{equation}
F= \kappa {\bf e} g_1 (R_w(X) +w)-((R_v(X) + \hat {\bf v}) \cdot \nabla ) (R_v(X) + \hat {\bf v}) ,
\label{F2}
\end{equation}
\begin{equation}
G=  \eta_1 - ((R_v(X) + \hat {\bf v}) \cdot \nabla ) (R_w(X) + \hat w + u_1), 
\label{G2}
\end{equation}
\begin{equation}
F_i=\langle  F, {\bf v}_i^*\rangle + \langle G,  \theta_i^*\rangle.
\label{F2}
\end{equation}

We consider   equations (\ref{X2}), (\ref{hom2}) and (\ref{hw2})  in the domain
\begin{equation}
  {\mathcal W}_{\gamma, R_0, C^{(1)}, \alpha} =\{(X, \hat w, \hat {\bf v}): \ |X| <  R_0, \ ||\hat {\bf v}||_{\alpha} + 
   ||\hat w||_{\alpha}  < C^{(1)}\gamma \}.
\label{Dom}
\end{equation}
That domain is a tubular neighborhood of the ball ${\mathcal B}^N(R_0)$ and the parameter $C^{(1)}>0$ is independent of $\gamma$ for small $\gamma$.  The width of that
neighborhood is $C^{(1)} \gamma$.

\begin{lemma} \label{7.1}
{
Let $\delta  \in (0,1)$ and $\alpha \in (1/2, 1)$.
Assume $\gamma > 0$ is small enough:
$
\gamma < \gamma_0(N,  R_0,  b, \alpha, \delta, C^{(1)}).
$
Then the local semiflow $S^t$, defined by equations  (\ref{X2}),(\ref{hom2}), and (\ref{hw2}) has a locally invariant in the set ${\mathcal W}_{\gamma, R_0, C, \alpha} \subset {\mathcal H}$  
and locally attracting manifold ${\mathcal M}_{N}^{(1)}$ of dimension $N$. This manifold is  defined by
\begin{equation}
 \hat {\bf v}= \hat {\bf v}_0(X, \gamma),   \quad \hat w= 
   \hat w_0(X, \gamma),
\label{rW}
\end{equation}
where
$\hat {\bf v}_0(X, \gamma)$, $\hat w_0(X,\gamma)$ are maps from the ball ${\mathcal B}^{N}(R_0)$ to $H_{\alpha}$ and $\tilde H_{\alpha}$, respectively.
They are bounded in $C^{1+\delta}$ -norm :
\begin{equation}
  |\hat {\bf v}_0(X, \gamma)|_{C^{1+\delta}({\mathcal B}^{N}(R_0))} < C_1\gamma, 
\label{rWe}
\end{equation}
\begin{equation}
   |\hat w_0(X, \gamma)|_{C^{1+\delta}({\mathcal B}^{N}(R_0))} < C_2 \gamma,  
\label{rWe1}
\end{equation}
constants $C_i >0$ are uniform in $\gamma$. The restriction of  the semiflow $S^t$ on ${\mathcal M}_{N}$  is defined by the system of differential equations
\begin{equation}
  \frac{dX_i}{dt}=\gamma (V_{i}(X)  +  \tilde V_i(X, \gamma)), 
\label{maineq1}
\end{equation}
where    
\begin{equation}  \label{Vpm}
V_i(X)= F_i(X, 0, 0)
\end{equation}
and the corrections $\tilde V_i(X, \gamma)=
 F_i( X, \hat {\bf v}_0(X,\gamma),\hat w_0(X, \gamma))   - F_i(X, 0,0)
$ satisfy the estimates
\begin{equation}
  |\tilde V_i|, |D_X \tilde V_i|  < c_1 \gamma^s,   \quad  s>0.
\label{maineq4}
\end{equation}
}
\end{lemma}

This assertion is proved in  Appendix.

Due to Theorem on persistence of hyperbolic sets  \cite{Katok}, for sufficiently small $\gamma$  we can remove  small corrections $\tilde V_i$ in the right hands of  (\ref{maineq1}).    Then, after a time rescaling,  we obtain   from \eqref{maineq1}  the following system of differential equations with quadratic nonlinearities:  
\begin{equation}
   \frac{dX_i}{dt}=V_i(X)=   { K_i}(X)    + { M_i}(X) + f_i,
\label{gks}
\end{equation}
which does not involve the small parameter $\gamma$.
Let us note that $K_i$ and $M_i$  can be represented  as
\begin{equation}\label{Gpm}
  K_i( X) =\sum_{j,l=1}^N K_{ijl} X_j X_l,  \quad M_i( X) = \sum_{j=1}^N M_{ij}X_j.
\end{equation}

Using the stream- function representation of the eigenfunctions ${\bf v}_j$  we see that the coefficients $K_{ijl}$ and $M_{ij}$ in (\ref{Gpm}) can be computed  by the relations
\begin{equation}
   M_{ij}(u_1(\cdot, \cdot))=
\langle \{\psi_j,  \theta_i^*  \},  u_1\rangle,
\label{MatrM}
\end{equation}
and for large $\nu$
\begin{equation}
   |K_{ijl}-\langle \{\psi_j, \theta_l\},  \theta_i^{*} \rangle| < c_0\nu^{-1}.
\label{G+}
\end{equation}
  Here we have used estimate (\ref{estz}), which implies that the scalar products,  where the fluid components ${\bf v}_j^*$ of the conjugate
eigenfunctions $e_j^*$ are involved, have the order $O(\nu^{-1})$.

\section{Control of linear terms  in system (\ref{gks})}
\label{sec: 10}

In this section, we  show that the coefficients $M_{ij}$ involved in system  (\ref{gks})
are completely controllable by the function $u_1(x,y)$.

To calculate the entries of $M_{ij}$ 
  we use the relation
\begin{equation}
   M_{ij}(u_1(\cdot,\cdot))=
\langle \{\psi_j,  \theta_i^*  \},  u_1\rangle.
\label{M}
\end{equation}
 Using     \eqref{eig1}, \eqref{eig1a},  \eqref{eig2}, \eqref{eig1m} and \eqref{eig3m} one obtains
\begin{equation}
   M_{ij}(u_1(\cdot,\cdot))=\frac{1}{2}\int_0^{2\pi} \int_0^h
[\tilde \zeta_{i j}(y)\cos((k_i +k_j)x)
+  \zeta_{i j}\cos((k_i - k_j)x)]
 u_1(x,y)dx dy,
\label{M++}
\end{equation}
where
$$
\zeta_{ i j}=k_j\Psi_{j}(y)\frac{d \Theta_{i}^*(y)}{dy} + k_i \frac{d\Psi_{j}(y)}{dy} \Theta_{i}^*(y),
$$
$$
\tilde \zeta_{i j}=k_j\Psi_{j}(y)\frac{d \Theta_{i}^*(y)}{dy} - k_i \frac{d\Psi_{j}(y)}{dy} \Theta_{i}^* (y).
$$
Using relations \eqref{eig1a},  \eqref{eig2} and \eqref{eig3m}, one obtains
\begin{equation}
   \tilde \zeta_{i j}= \bar a_i \bar b_j \tilde \eta_{ij}, \quad \zeta_{i j}= \bar a_i \bar b_j \eta_{ij},
\label{M2}
\end{equation}
where $\bar a_i, \bar b_j$ are some non-zero coefficients, and 
\begin{equation} 
    \eta_{i j}=y^2 (k_j+ 2k_i   + 2 k_i^2 y - 2k_i^2 k_j y^2) \exp(-(k_i +k_j) y)  + O(b^{-r_0}),  
\label{eta1}
\end{equation}
\begin{equation}
\tilde \eta_{i j}=y^2 (k_j - 2k_i  +2k_i(k_j -k_i)y) \exp(-(k_i +k_j) y) +  O(b^{-r_0})
\label{eta2}  
\end{equation}
for some $r_0 >0$.

\begin{lemma} \label{8.1}
{
For each   $N \times N$  matrix $T$ and each $\delta >0$  there
exists a $2\pi$ -periodic in $x$ smooth function $g_1(x,y)$
 such that for sufficiently large $b$ one has
\begin{equation}
 | M_{jl}(u_1(\cdot,\cdot))  - T_{jl}| < \delta, \ \forall j,l=1,...,N,
\label{contM1}
\end{equation}
where $u_1$ is defined via $g_1$ by  \eqref{Uu1}.

}
\end{lemma}

\begin{proof}
We are looking for $u_1$ satisfying \eqref{contM1}. Just $u_1$ is found we use relation \eqref{Uu1} and define
$g_1$ as
$$
g_1=\frac{u_1}{ (U-u_0) (1 + \gamma u_1)}
$$
choosing a sufficiently large $u_0>0$ such that the function $U -u_0$ has no roots.
 
We represent $u_1(x,y)$ by  a Fourier series:
$$
u_1(x,y)=\hat u_0(y) + \sum_{k=1}^{+\infty} \hat u_{k}(y) \cos( k x). 
$$
Then relation  (\ref{M++}) gives
\begin{equation}
   M_{ij}(u_1)=\frac{1}{2} \int_0^h
\big (\tilde \zeta_{i j}(y)\hat u_{k_i +k_j}(y) 
+  \zeta_{i j}\hat u_{|k_i -k_j|}(y) \big)  dy.
\label{M2++}
\end{equation}
 We introduce the auxiliary quantities
$V_{n,m}^{(p)}$  by
\begin{equation}
   V_{n,m}^{(p)}= \int_0^h  y^{2+p}  \exp(-my)  \hat u_n(y) dy,
\label{Vpm}
\end{equation}
 where $n>0, m, p \ge 0$ are integer indices. 
Let us prove an auxiliary assertion.

\begin{lemma} {For any  $a_{m, p}$, where $m=1,..., M$ and $p=0,1,2$,  we can find a  function $W(y) \in C^2([0,h])$ such that
\begin{equation}
    \int_0^h  y^p \exp(-my) W(y) dy=a_{m, p},    \quad \forall  \ m=1,..., M, \ p=0,1,2.
\label{Vpm}
\end{equation}
and 
\begin{equation}
    W(0)=W(h)=W'(0)=W'(h)=0.
\label{Vpm2}
\end{equation}
}
\end{lemma}

\begin{proof} Consider the map $W(\cdot) \to a_{m,p}$ defined by \eqref{Vpm} on the space $E_3$
of  the functions $\in C^2([0,h])$ and satisfying \eqref{Vpm2}. The range of that map  is a  linear subspace
of ${\mathbb R}^{3M}$.  If this assertion is not fulfilled, there is a vector orthogonal to the closure of that linear subspace and thus there exist numbers $b_{m,p}$  such that 
$$\sum_{m=1}^M \sum_{p=0}^2 |b_{m,p}|=1$$
and 
$$
\sum_{m=1,..., M} \sum_{p=0,1,2} \int_0^h  b_{m,p} y^p \exp(-my) W(y) dy=0
$$
for all  $W(y) \in C^2([0,h])$ such that \eqref{Vpm2} holds.  This means that
$$
\sum_{m=1,..., M} \sum_{p=0,1,2}   b_{m,p} y^p \exp(-my) = 0, \quad \forall \ y \in (0,h),
$$
i.e., the  functions $y^p \exp(-my)$ are linearly dependent that is not the case. \end{proof}

Using that lemma we assume that 
all $V_{n,m}^{(p)}=0$ for all $n=m=k_i +k_j$.
Then by \eqref{M2++} and \eqref{Vpm} estimate   (\ref{contM1}) can rewritten as follows:
\begin{eqnarray} \label{M3}
 |(k_j +2k_i)V^{(0)}_{|k_j -k_i|, k_j +k_i} +  2k_i^2
V^{(1)}_{|k_j -k_i|, k_j +k_i}  -  \nonumber \\
-2 k_j k_i^2  V^{(2)}_{|k_j -k_i|, k_j +k_i} +  \tilde B_{ij}(b)-  \bar T_{ij}| <\delta,
\end{eqnarray}
where $\tilde B_{ij}$, \ $i,j=1,..., N$ satisfy 
$
|\tilde B_{ij}| <  c_1 b^{-c_{2}}, \quad c_{1}, c_2 >0
$
and  $ \bar T_{ij}= T_{ij}  \tilde a_{i} \tilde b_j$, where $\tilde a_i, \tilde b_j$ are coefficients.

Estimate \eqref{M3} shows that, in order to prove the assertion,  it suffices to find 
$X_{n, m}=V^{(0)}_{n, m} $ and $Y_{n,m}=V^{(1)}_{n, m} $ satisfying the system
\begin{equation} \label{VV03}
  (k_j + 2k_i) X_{|k_j -k_i|, k_j +k_i}   
+ 2k_i^2   Y_{|k_j -k_i|, k_j +k_i} =\bar T_{ij}, 
\end{equation}
where $\ i, j\in {1,..., N}$.

We decompose that system into the symmetric and antisymmetric parts.  Then we obtain
\begin{equation} \label{VV03s}
   \frac{3}{2} (k_j + 2k_i) X_{|k_j -k_i|, k_j +k_i} +  
 (k_i^2 +k_j^2)  Y_{|k_j -k_i|, k_j +k_i} =\bar T_{ij}^{s},  i \ge j, 
\end{equation}
\begin{equation} \label{VV03a}
  \frac{1}{2} (k_i - k_j) X_{|k_j -k_i|, k_j +k_i} +  
 (k_i^2 - k_j^2)  Y_{|k_j -k_i|, k_j +k_i}  =\bar T_{ij}^{a},   \quad  i > j, 
\end{equation}
for some $\bar T_{ij}^{s}$ and $\bar T_{ij}^{a}$, where $\ i, j\in {1,..., N}$.
Consider the map $R_N: (i,j) \to |k_j -k_i|, k_j +k_i$ defined on the set of the pairs $(i,j)$
where $i=1,..., N$, $j=1,..., N$ and $i \ge j$.  That map  is  an injection, therefore,
  the  system of equations \eqref{VV03s} and \eqref{VV03a} can be represented as a set of  independent systems of two linear equations for two unknowns.  For each $(i,j)$ the corresponding linear $2 \times 2$ system is resolvable that can be checked by
the determinant calculation.\end{proof}

Let us formulate a lemma about control $f$ by $\eta_1$.

\begin{lemma} \label{fcontrol}
{ Given a vector $f=(f_1, ..., f_{N})$, there exists a
 smooth $2\pi$-periodic in $x$  function $\eta_1(x, y)$ such that
$$
\langle  \tilde \theta_i, \eta_1 \rangle =  f_i, \quad i=1,..., N.
$$
}
\end{lemma}
We omit an elementary proof.

In coming sections we investigate system  (\ref{maineq1}) mainly following works \cite{Sud} and \cite{Stud,Vak4}. 

\section{Quadratic systems}
\label{sec:5}

System (\ref{gks}) defines a local semiflow $S^t(f,{ M})$ in the ball ${\mathcal B}^N(R_0) \subset {\mathbb R}^N$ of the radius $R_0$
centered at $0$. We shall consider the vector $f$ and the matrix $ M$ as parameters of this semiflow whereas
the entries $K_{ijl}$ will be fixed.

Let us formulate an assumption on entries $K_{ijl}$.
We  represent  $X$ as a pair $X=(Y, Z)$, where
$$
Y_l=X_{l},  \quad l \in I_p, \quad Z_j=X_{j+p}, \quad j \in J_p,
$$
where  $I_p=\{1,..., p \}$ and $J_p=\{1, ..., N-p \}$.
Then  system (\ref{gks}) can be rewritten as
 \begin{equation}
   \frac{dY}{dt}=   { K}^{(1)} (Y) + { K}^{(2)} (Y, Z) + { K}^{(3)}(Z) + { R} Y + { P} Z
 +  f,
\label{gks2}
\end{equation}
 \begin{equation}
   \frac{dZ}{dt}=  \tilde { K}^{(1)} (Y) + \tilde{ K}^{(2)} (Y, Z) + \tilde { K}^{(3)} (Z)  +  \tilde { R} Y +  \tilde { P} Z
 +  \tilde f,
\label{gks3}
\end{equation}
where  for $i=1,..., p$ 
\begin{equation}
   { K}^{(1)}_i(Y)=  \sum_{j \in I_p} \sum_{l \in I_p}  K_{ijl}^{(1)} Y_{j} Y_{l}, \quad { K}^{(3)}_i(Z)=  \sum_{j \in J_p} \sum_{l \in J_p}  K_{ijl}^{(3)} Z_{j} Z_{l},
\label{KY1}
\end{equation}
\begin{equation}
   { K}^{(2)}_i(Y,Z)=  \sum_{j \in I_p} \sum_{l \in J_p}  K_{ijl}^{(2)} Y_{j} Z_{l}, 
\label{KY2}
\end{equation}
and for $k=1,..., N-p$
\begin{equation}
    \tilde { K}_k^{(1)}(Y)= \sum_{j\in I_p} \sum_{l \in I_p}  \tilde K_{kjl}^{(1)} Y_{j} Y_{l}, \quad \tilde { K}_k^{(3)}(Z)=\sum_{j \in J_p} \sum_{l \in J_p}  \tilde  K_{kjl}^{(3)} Z_{j} Z_{l},
\label{KY3}
\end{equation}
\begin{equation}  \label{K2YZ}
\tilde { K}_k^{(2)}(Y,Z)= \sum_{j \in I_p} \sum_{l \in J_p}  \tilde  K_{kjl}^{(2)} Y_{j} Z_{l}.
\end{equation}
Note that
\begin{equation}
     \tilde K_{kjl}^{(1)} =K_{k+p, jl},  \quad k=1,...,p,  \quad  j,l =1,..., p.
\label{KGBfollowsyou}
\end{equation}

The linear terms $MX$ take the form 
 \begin{equation}
 ({ R} Y)_i= \sum_{j \in I_p} R_{ij} Y_j,  \quad  (\tilde { R} Y)_k= \sum_{j \in I_p} \tilde R_{kj} Y_j,
\label{gks3c}
\end{equation}
\begin{equation}
  ( { P} Z)_i=\sum_{j \in J_p}  P_{ij} {Z_j},
 \quad  ({\tilde  P} Z)_k= \sum_{j \in J_p} \tilde P_{kj} Z_j,
\label{gks3d}
\end{equation}
and $f=(f_{1}, ..., f_{p}), \ \tilde f=(\tilde f_1, ..., \tilde f_{N-p})$.

We denote by  $S^t({\mathcal P})$ the local semiflow defined by
(\ref{gks2}) and (\ref{gks3}). Here ${\mathcal P}$ is a semiflow parameter, $ {\mathcal P}=\{f, \tilde f, {P}, \tilde { P},
{R}, \tilde { R} \}$.  Let us formulate an assumption on quadratic terms $K_i(X)$. 
\vspace{0.2cm}

{\bf $p$-Decomposition Condition} \label{ass40}
{\em
Suppose  entries $K_{ijl}$ satisfy  the following condition.
For some $p$ there exists a decomposition $X=(Y, Z)$, where $Y \in {\mathbb R}^p$ and $Z \in {\mathbb R}^{N-p}$ such that for all   $b_{jl}$
 the  linear system  
\begin{equation}
	     \sum_{i \in J_p} \tilde K_{ijl}^{(1)}  \chi_i= b_{jl}, \quad l, j \in I_p
\label{barK}
\end{equation}
has a solution $\chi$.
}
\vspace{0.2cm}

Clearly that for $N > p^2 + p$ and  generic matrices $K$ this condition is valid.

Let us formulate  some conditions to  the matrices ${ R}, \tilde { R}, {P}$ and $\tilde { P}$.
Let $\xi >0$ be  a  parameter.  We suppose that
\begin{equation}
  \tilde P_{ij}=-\xi^{-1} \delta_{i,j},    \quad  i=1,..., N-p, \ j=1, ...,
\label{gks4c}
\end{equation}
\begin{equation}
  \tilde R_{ij}=0, \quad  \tilde f_i=0, \quad i=1,..., N-p, \ j=1,...,p,
\label{gks4d}
\end{equation}
\begin{equation}
   P_{ij}=\xi^{-1} T_{ij}, \quad |T_{ij}| < C_0, \quad i=1,..., p, \ j=1,..., N-p,
\label{gks5c}
\end{equation}
\begin{equation}
 |R_{ij}| < C, \quad i=1,..., p, \ j =1,..., p,
\label{gks5d}
\end{equation}
Let us define the domain in ${\mathbb R}^N$:
\begin{equation} \label{domainR0}
{\mathcal  W}_{\bar R, C^{(2)},\xi}=\{X=(Y, Z):   \quad  |Y| < \bar R,\ |Z| < C^{(2)}\xi   \},  \quad  C^{(2)}>0.
\end{equation}
Note that ${\mathcal  W}_{\bar R,C^{(2)},\xi}$ is a tubular neigborhoof of the ball ${\mathcal B}^p(\bar R)$
of the small width $C^{(2)}\xi$.

\begin{lemma} \label{quadr1} {Assume  (\ref{gks4c}),  (\ref{gks4d}), (\ref{gks5c}) and (\ref{gks5d})  hold and $\bar R >0, c_0>0$ are constants.
For sufficiently small positive $\xi < \xi_0(\bar R, \delta, M,K,f, p, N, c_0)$ and $C^{(2)} >c_0$ the local semiflow $S^t({\mathcal P})$ defined by  system (\ref{gks2}), (\ref{gks3}) 
has a locally invariant in the domain ${\mathcal  W}_{\bar R, C^{(2)},\xi}$
and locally attracting manifold $\tilde {\mathcal M}_p^{(2)}$.  This manifold is  defined by equations
\begin{equation}
  Z= \xi (\tilde { K}^{(1)}(Y)  +  W(Y,\xi)), \quad Y \in {\mathcal B}^p(\bar R)
\label{gks6c}
\end{equation}
where $W$ is a $C^{1+\delta}$ smooth map defined on the ball ${\mathcal B}^p(\bar R)$ to ${\mathbb R}^{N-p}$ and such that
\begin{equation}
|W(\cdot, \xi)|_{C^1({\mathcal B}^p(\bar R))}  < C_1\xi^s, \quad s > 0.
\label{gks7c}
\end{equation}}
\end{lemma}

{\bf Proof} can be found in the paper \cite{Sud}.

The semiflow $S^t$ restricted to $\tilde {\mathcal M}_p$ is defined by the equations
\begin{equation}
\frac{dY}{d\tau}=   \xi S(Y,  \xi),
\label{inert}
\end{equation}
where
$$
    S(Y,  \xi)={ K}^{(1)} (Y) + \xi { K}^{(2)} (Y, \tilde { K}^{(1)}(Y)  +  W(Y, \xi)) +
$$
$$
    + \xi^2 { K}^{(3)}(\tilde { K}^{(1)}(Y)  +  W(Y,\xi)) + { R} Y + { T}\tilde { K}^{(1)}(Y)  + W(Y, \xi))
 +  f.
$$
The estimates for $W$ show that $S$ can be presented as
\begin{equation}
S(Y,  \xi)={ K}^{(1)} (Y)  + { R} Y + { T}\tilde { K}^{(1)}(Y) + f + \tilde S(Y, \xi)
\label{inert2}
\end{equation}
where a small correction $\tilde S(Y, \xi)$ satisfies
\begin{equation}
|\tilde S(Y, \xi)|_{C^1({\mathcal B}^p(R_0))} < c_0\xi^{1/2}.
\label{inert2c}
\end{equation}
In (\ref{inert2})  $ R$ and $f$ are free parameters. The quadratic form  ${ D}(Y)={ K}^{(1)} + { T}\tilde { K}^{(1)}$
 can be also considered as  a free parameter according to $p$- Decomposition Condition.
Therefore, we have proved the following assertion.

\begin{lemma} \label{quadr2} { Let
\begin{equation} \label{QField}
W(Y)= { D}(Y) + { R} Y + f
\end{equation}
be a quadratic vector field on the ball ${\mathcal B}^p(\bar R)$, $\bar R >0$, where
$$
{D}_i(Y)=\sum_{j=1}^p \sum_{l=1}^p  D_{ijl} Y_j Y_l, \quad ({ R} Y)_i=\sum_{j=1}^p R_{ij} Y_j.
$$
Consider system  (\ref{gks2}), (\ref{gks3}). Let $p$- Decomposition Condition hold.
Then for any $\epsilon >0$  and $R_d >0$ the field $W$ can be $\epsilon$ - realized by local  semiflow  defined  by
system (\ref{gks2}), (\ref{gks3}) on $C^{1+\delta}$-smooth locally invariant in
a  neighborhood ${\mathcal  W}_{\bar R, C^{(2)}, \xi}$ of ball ${\mathcal B}^p(\bar R)$ and locally attracting manifold ${\mathcal M}^{(2)}_p$ of dimension 
$p$.  Here parameters $\mathcal P$ are the matrices ${ P}$, $ R$, $ \tilde  P$,
$ \tilde  R$ and the vectors $f, \tilde f$.
}
\end{lemma}

 Lemma \ref{quadr2}  and results \cite{Stud} imply the following assertion.  Let us consider the families $\Phi_{2, R_0}$ of quadratic fields $V(X)$ defined 
on the ball  ${\mathcal B}^N(R_0)$ by \eqref{gks} and depending on a parameter ${\mathcal P}$
as follows.   Each field in a family $\Phi_{2, R_0}$ is defined by numbers $p, N$, the matrix $M$, the coefficients $K_{ijl}$, where $i,j,l \in {1,...,N}$,  and the vector $f$.  We consider $p, N$, the matrix $M$ and the vector $f$ as free parameters, i.e., the parameter $\mathcal P$ of our family  is a quadruple $\{p, N, M, f \}$, where $p, N$ runs over the set of all positive integers  ${\mathbb N}_{+}$, 
and $p^2 + p < N$.  For each fixed $N$ the  range of the parameter $M$ is the set of all square $N \times  N$ matrices and the range of  $f$  is  ${\mathbb R}^N$.
For each ${\mathcal P}$ the corresponding coefficients $K_{ijl}$  
are defined uniquely and  satisfy $p$-Decomposition condition.

\begin{proposition} \label{QuadrP} { Consider a family of the semiflows defined by  systems  (\ref{gks}),   where $V \in \Phi_{2, R_0}$. Then that family  enjoys the following property. 
For each integer $n$, each $\epsilon > 0$  and each vector field
$Q$ satisfying (\ref{cond1}) and (\ref{inward}), there exists a value of the parameter
${\mathcal P}={\mathcal P}(Q, \epsilon, n)$ such that
the corresponding system   (\ref{gks}) defines a  semiflow, which $\epsilon$-realizes
the vector field $Q$.
 } 
\end{proposition}

\subsection{ Verification of $p$-Decomposition condition for system (\ref{maineq1})}
\label{sec: 10.2}

Let us fix a $p$. To verify $p$-Decomposition condition for system
(\ref{gks}) we choose the set ${\mathcal K}_N$ from Prop. {6.2}  in a special way.  Namely, we set
$ {\mathcal K}_N={\mathcal K}_{p, N}$, where the set ${\mathcal K}_{p, N}$ is defined below as follows.

Let us denote by $P_{2,p} $ the set of non-ordered pairs $(i,j)$, where  $i, j \in \{1,..., p\}$. The equality $(i,j)=(i', j')$ means that
either  $i=i', j=j'$  or $i=j', j=i'$. Let ${\mathcal S}_p$ be the set consisting of all sums $k_i +k_j$, where $i, j \in \{1,..., p\}$.
Consider the map $S_p: P_{2,p} \to {\mathcal S}_p $ from the set 
$P_{2,p}$
on the set ${\mathcal S}_p$  defined by $S_p((i,j))= k_i +k_j$. Let us prove an auxiliary lemma.

\begin{lemma} \label{8.3}
{ For each $p>1$ there exists a set $\bar {\mathcal K}_p =\{\bar k_1,..., \bar k_p\}$ of  integers $\bar k_i>0$ such that 
\begin{equation} \label{5555}
  \bar k_j \ne 5n,  \quad \forall n \in {\mathbb N}
\end{equation}
and 
all the sums $\bar  k_i + \bar k_j$ are mutually distinct, i.e.,
\begin{equation} \label{sum2}
\bar  k_i + \bar  k_j=\bar  k_{i'} + \bar  k_{j'}    \implies (i, j)=(i', j'). 
\end{equation}
 In the other words,  the map $S_p$ is injective.   
}
\end{lemma}

\begin{proof}
The set $\bar {\mathcal K}_p=\{\bar  k_1,  \bar  k_2, ..., \bar  k_p\}$  can be found by an induction.
For $p=2$ we set $\bar  k_1=1, \bar  k_2=7$. Suppose   $\bar {\mathcal K}_p$ is found for some $p$. Then we  take an odd $k_{p+1}$  such that
$\bar  k_{p+1} >  \bar  k_{j_1} + \bar  k_{j_2}$ for all $j_1, j_2  \in \{1, ..., p\}$.
Then the extended set  $\bar { \mathcal K}_{p+1}=\{\bar  k_1, ..., \bar  k_p, \bar  k_{p+1}\}$ satisfies conditions
(\ref{sum2}). 

To show it, consider two sums from (\ref{sum2}).   If  the pairs  $(i, j)$ and $(i', j')$
do not include the index $p+1$ then $(i, j)=(i', j')$ by the induction assumption.  If
one of $i, j$ equals $p+1$ but the pair $(i', j')$ does not include $p+1$,  equality (\ref{sum2}) is not fulfilled that follows from the construction of $k_{p+1}$.  
Therefore, the both pairs $(i, j)$ and $(i', j')$  contains $p+1$.  Then we can exclude this index, that
gives either $i=i'$ or $i=j'$. \end{proof}

For all $p$ we define the sets ${\mathcal K}_{p, N}=\{ k_1,   k_2, ..., k_N\}$, where $N=p(p+1)/2$, as follows. Let us take  the set $\bar {\mathcal K}_{p}=\{\bar k_1,...,\bar k_p \}$  satisfying the conclusion of lemma \ref{8.3}.  For $ i \le p$ we take $ k_i=\bar  k_i$. The integers $k_i \in {\mathcal K}_{p, N}$ with $i >p$ we define as different sums from $\bar {\mathcal K}_{p}$, i.e., $ k_i=\bar  k_{i_1} + \bar  k_{i_2}$ for some $ i_1, i_2 \in \{1,..., p\}$.

 Let us set
$Y_l= X_{l}$,  $l=1,...,p$.
Respectively, all the rest variables  $X_j$ with $j > p$ will be $Z_l$, where $l=j-p$.
 Then, due to relation \eqref{KGBfollowsyou}, to verify $p$-Decomposition condition \ref{ass40}, 
it is sufficient to  check that the linear system 
\begin{equation} \label{eqG}
\sum_{i=p+1}^{N} K_{ijl} \chi_i  =b_{jl}, \quad j, l \in I_p=\{1,...,p\}
\end{equation}
 has a solution for any given $b_{jl}$.

To verify it, let us calculate
the coefficients $K_{ijl}$ defined by  (\ref{G+}).
 Integrating by parts one has
\begin{equation}
   K_{ijl}=-\langle \{\psi_j,  \theta_i^* \}, \theta_l \rangle + O(\nu^{-1}).
\label{GM11}
\end{equation}
Thus  by definition (\ref{MatrM}) 
\begin{equation}
   K_{ijl}=-M_{ij}( \theta_l(\cdot, \cdot))  + O(\nu^{-1}).
\label{GM12}
\end{equation}
Using that relation and (\ref{M++})
one obtains
\begin{equation}
   K_{ijl}=\frac{1}{4}\big(
       \delta_{k_{i}, k_l+ k_j}  I_{ijl}
    + O(b^{-c_1})\big), \quad c_1 >0,
\label{GM++2}
\end{equation}
and
$$
  I_{ijl} =\int_0^h  \tilde \zeta_{ij}(y)  \Theta_{l}(y) dy,   \quad p+1 \le i \le N, \  j, l=1,..., p.
$$
Relation \eqref{GM++2} means that for each fixed pair $(j,l)$ the sum in the left-hand side of system \eqref{eqG} consists of a single term with the index $i$ defined  by 
$k_i=\bar k_j +  \bar k_l$. Due to Lemma  \ref{8.3}   system  \eqref{eqG} can be decomposed in independent linear equations, each of them involves only a
single unknown $\chi_i$.
Thus system \eqref{eqG} is resolvable under condition that all coefficients $ I_{ijl}$ 
are not equal $0$. To compute those coefficients  we take into account relation
 (\ref{eta2}) for $\tilde \xi_{ij}$.  For large  $b$ we find that
\begin{equation}
  I_{ij l}=\bar a_i \bar b_j \beta_l \Big(J(\bar k_j, \bar k_l) +  O(b^{-r_2}) \Big) \delta_{k_i, \bar k_l+ \bar k_j}, \quad r_2 >0,
 \label{Iij}
\end{equation}
where $\bar a_i, \bar b_j,  \beta_l \ne 0$ and 
$$
J(\bar k_j, \bar k_l) =2(2 (\bar k_l+ \bar k_j))^{-3} (k_j - 5k_l).
$$
We note  \eqref{5555} implies that $J(\bar k_j, \bar k_l)  \ne 0$ for all integers $j,l=1,..., p$. 
Thus for each integer $p$ 
we can solve  system (\ref{eqG}) and the p-Decomposition condition is  fulfilled.

\section{Proof of Theorems}
\label{sec: 11}

\begin{proof} of Theorem \ref{maint}.
 Let $\epsilon>0$. We suppose  that a vector field $Q$ defined on the ball ${\mathcal B}^n$ satisfy  (\ref{cond1}) and (\ref{inward}).
Our goal is to find parameters $\mathcal P$
of IBVP  (\ref{OB1}) -(\ref{Maran1}) 
such that the corresponding family of  semiflows, generated by that IBVP,   
$\epsilon$ -realizes $Q$. 
\vspace{0.2cm}

{\em Step 1}. 
According to Proposition \ref{QuadrP},  for  each $\epsilon_0 > 0$ we can $\epsilon_0$- realize the field $Q$  by a semiflow defined by a quadratic vector field  
$V(X)$ from a family  $\Phi_{2, R_0}$.  The field $V$ is defined 
 on a ball ${\mathcal B}^N \subset {\mathbb R}^N$.

{\em Step 2}.  Consider the family ${\mathcal F}_{OB}$ of global semiflows defined by 
by  IBVP (\ref{OB1}) -(\ref{Maran1}) with parameters
${\mathcal P}=\{h, \nu, \gamma, \beta, \beta_1, u_0, \eta(\cdot,\cdot), g_1( \eta(\cdot,\cdot) \}$.  For any $\epsilon_1>0$ that family $\epsilon_1$-realizes (in the sense of definition \ref{RVFmax})
a family of  semiflows $\Phi_{2, R_0}$ considered at the previous step.  Therefore, if $\epsilon_0, \epsilon_1$ are small enough,
the family ${\mathcal F}_{OB}$   realizes $Q$ with accuracy $\epsilon$.

\end{proof}

\begin{proof}  of of Theorem \ref{maint2}.

Consider a  global semiflow on finite dimensional  smooth compact manifold defined by a $C^1$-smooth vector field and having a  hyperbolic compact invariant set $\Gamma$.  
For an integer $n>0$ we can find a smooth vector field $Q$ on a unit ball ${\mathcal B}^n$, which generates a semiflow having a topologically equivalent hyperbolic compact invariant set $\Gamma'$
(and the corresponding restricted dynamics are orbitally topologically equivalent) .  
Due to the Theorem  on Persistence of Hyperbolic sets (see \cite{Ru, Katok})  
we find a sufficiently small $\epsilon(\Gamma', Q) >0$ such that for all 
$C^1$ perturbations  $\tilde Q$ of $Q$ satisfying $|\tilde Q|_{C^1({\mathcal B}^n)} <\epsilon$ the perturbed systems 
\begin{equation} \label{QQq}
dq/dt=Q(q)+  \tilde Q(q)
\end{equation}
have
hyperbolic compact invariant sets $ \tilde \Gamma$   topologically equivalent  to $\Gamma$ (and the corresponding restricted dynamics are orbitally topologically equivalent). Then we $\epsilon$- realize this field by Theorem \ref{maint}.  

Now, to finish proof, it is sufficient to prove that trajectories defined by system \eqref{QQq}
on ${\mathcal B}^n$ 
do not leave the locally invariant manifold ${\mathcal M}_n$.   By definition of locally invariant manifolds (see Sect. \ref{sec:3}), it suffices to prove that those trajectories do not leave
the corresponding domain $\mathcal W$, where that manifold is locally invariant.  For sufficiently small $\epsilon>0$  the trajectories $q$ are bounded. Indeed, the corresponding perturbed field $Q(q) + \tilde Q(q)$, defined on ${\mathcal M}_{n}$, directed  inward the ball $B^n$ at the boundary ${\partial B}^n$ and thus the corresponding semitrajectories do not leave that ball.  
The corresponding trajectories $z(t)=(X(t), \tilde {\bf v}, w(t))$ of semiflow $S^t$ generated by our IBVP on the manifold 
${\mathcal M}^{(1)}_N$ also are bounded. We choose  radius $R_0$ and the width $C^{(1)}$ such that  the domain ${\mathcal W}_{\gamma, R_0, C^{(1)}, \alpha}$
contains these trajectories $z(t)$. We make an analgous choice for all locally invariant manifolds ${\mathcal M}^{(2)}_p,  {\mathcal M}^{(3)}, ... $ involved in our realization,
adjusting  $R_d$,$\bar R$ and the corresponding width parameters $C^{(2)}$ and $C^{(3)}$. Then $z(t)=(X(t), \tilde {\bf v}, w(t))$ do not leave ${\mathcal M_n}$.  It finishes the proof.

\end{proof}


\section{Conclusion}
\label{sec: 12}

The idea that a complicated behaviour of  dissipative dynamical systems,  associated with
 fundamental models of physics, chemistry and biology can be generated by a strange (chaotic) attractor was pioneered
in the seminal work of D. Ruelle and F. Takens \cite{RT}. In this paper, it is shown that
classical system of hydrodynamics, which appears in many applications,  can exhibit all kinds of structurally stable chaotic behaviour. 
 The mathematical method  admits a  transparent  physical interpretation:   complicated large time dynamics can be produced by  an exponentially decreasing at the no-flip boundary temperature profile and small space inhomogeneous perturbations  of the gravity force and that profile. In this paper,  the complete analytical description of  complex turbulent patterns is given.

\section{Acknowledgements}

I dedicate this paper to  the memory of my friend Vladimir Shelkovich.

This work  was financially supported by Government of Russian Federation, Grant 074-U01.


\section{Appendix}

{\bf Proof of Lemma \ref{7.1}}.  This assertion is a consequence of Theorem 6.1.7  \cite{He}.
In the variables $\hat z=(\hat {\bf  v}, \hat w)^{tr}$, $X$ system (\ref{X2}), (\ref{hom2}), \eqref{hw2} can be rewritten as 
\begin{equation}
X_t= \gamma \hat F(X, \hat z),   
\label{eveq5a}
\end{equation}
\begin{equation}
\hat z_t=L \hat z  +  \hat G(X, \hat z).
\label{eveq6a}
\end{equation}
Using the standard truncation trick we modify  eq. (\ref{eveq5a}) as follows:  
\begin{equation}
X_t= \gamma \hat F(x, \tilde w)\chi_{R_0}(X),    
\label{eveq5at}
\end{equation}
where 
 $\chi_{R_)}$ is a smooth function such that $\chi_{R_0}(X) =1 $ for $|X| < R_0$ and  $\chi_{R_0}(X) =0 $ for $|X| >  2R_0$.
As a result of  that modification, $X$-trajectories of (\ref{eveq5at}) are defined for all $t \in (-\infty, +\infty)$ (as in Theorem 6.1.7 \cite{He}).
Then an invariant   manifold for the semiflow defined by system  (\ref {eveq5at}), \ref{eveq6a}) is a locally invariant one for the semiflow generated by
  (\ref {eveq5a}), \ref{eveq6a}). 

Let us consider the semigroup $\exp(Lt)$.  We have  estimates \eqref{semigroup1}, \eqref{semigroup2},
where $M, \bar M, \rho>0$ do not depend on $\gamma$. 
Moreover,
\begin{equation}
M_0=\gamma \sup_{(X, \hat z) \in {{\mathcal D}_{\gamma, 2R_0, C_1, C_2, \alpha}}} || \hat F \chi_{R_0} || <  c_2\gamma,
\label{Ppr2}
\end{equation}
\begin{equation}
\lambda=\gamma \sup_{(X, \hat z) \in {{\mathcal D}_{\gamma, 2R_0, C_1, C_2, \alpha}}}  ||D_X \hat F \chi_{R_0}|| + ||D_{\hat z}  \hat F \chi_{R_0}|| <  c_3\gamma,
\label{Ppr3}
\end{equation}
\begin{equation}
M_2=\gamma \sup_{(X, \tilde w) \in {{\mathcal D}_{\gamma, 2R_0, C_1, C_2, \alpha}}}  ||D_{\hat z}  \hat G|| <  c_4\gamma,
\label{Ppr4}
\end{equation}
We set   $\mu_0=\kappa/4$. Then for small $\gamma$
\begin{equation}
M_3=\gamma \sup_{(X, \hat z) \in {{\mathcal D}_{\gamma, 2R_0, C_1, C_2, \alpha}} } ||D_{X}  \hat G|| <  c_5\gamma.
\label{Ppr5}
\end{equation}
We set $\delta_1=2\theta_1$, where
\begin{equation}
\theta_p =\lambda M_0 \int_0^{\infty}  u^{-\alpha} \exp(-(\kappa - p\mu') u) du, \quad 1 \le p \le 1+\delta,
\label{Int1}
\end{equation}
and
$\mu'= \mu_0 + \delta_1 M_2$. For sufficiently small $\gamma$ one has $\mu' <\rho/2$, therefore, the integral in the right hand side
of (\ref{Int1}) converges and, according to (\ref{Ppr4}), one obtains $\theta < c_6 \gamma$ (since $M$ is independent of $\gamma$).
We notice then that for sufficiently small $\gamma$ the following estimates
$$
(1+\delta) \mu' < \rho/2,
$$
$$
\theta_1 < \delta_1(1 + \delta_1)^{-1} < 1, \quad  \theta_1(1+\delta_1)M_2{\mu'}^{-1} < 1,
$$
and
$$
\theta_p(1 +  \frac{(1+\delta_1)M_2}{r \mu'}) < 1
$$
hold.
Those estimates show that all conditions of Theorem 6.1.7  \cite{He} are satisfied, and Lemma \ref{7.1} is proved.

\end{document}